\documentclass[11pt]{article}
\usepackage{fullpage}
\usepackage[colorlinks=true,allcolors=blue]{hyperref}

\usepackage[T1]{fontenc}
\usepackage[utf8]{inputenc}
\usepackage{bbm}
\usepackage{graphicx}
\usepackage{subcaption}
\usepackage{floatrow}
\usepackage{amsmath}
\usepackage{amsfonts}
\usepackage{amsthm}
\usepackage{listings}
\usepackage{color}
\usepackage{microtype}
\usepackage{cleveref}
\usepackage{algorithm}
\usepackage{tcolorbox}
\usepackage{enumitem}
\usepackage[noend]{algpseudocode}
\usepackage{wrapfig}

\usepackage{lineno}
\definecolor{dkgreen}{rgb}{0,0.6,0}
\definecolor{gray}{rgb}{0.5,0.5,0.5}
\definecolor{mauve}{rgb}{0.58,0,0.82}

\newtheorem{invariant}{Invariant}[section]

\newtheorem{theorem}{Theorem}[section]

\newtheorem{lemma}[theorem]{Lemma}

\newtheorem{claim}{Claim}[section]

\newtheorem{observation}{Observation}[section]

\DeclareMathOperator*{\argmax}{arg\;max}
\DeclareMathOperator*{\argmin}{{arg\;min}}

\DeclareMathOperator*{\argminc}{arg\;minc}

\DeclareMathOperator*{\minc}{minc}
\DeclareMathOperator*{\mine}{mine}
\newcommand{\maxe}{\text{maxe}}

\newcommand{\eps}{\varepsilon}
\renewcommand{\epsilon}{\varepsilon}
\newcommand{\E}{\mathbf{E}}

\newcommand{\rank}{\mathrm{rank}}

\newcommand{\uni}{E}

\newcommand{\cI}{\mathcal{I}}
\newcommand{\cM}{\mathcal{M}}
\newcommand{\cN}{\mathcal{N}}
\newcommand{\cT}{\mathcal{T}}

\newcommand{\cD}{\mathcal{D}}

\newcommand{\cB}{\mathcal{B}}

\newcommand{\cS}{\mathcal{S}}

\newcommand{\RR}{{\mathbb R}}

\newcommand{\maxlight}{\text{maxlight}}
\newcommand{\minlight}{\text{minlight}}

\def\b1{{\bf 1}}

\def\bx{{\bf x}}

\newtcolorbox[auto counter, number within=section]{algbox}[2][]{%
    title=Algorithm~\thetcbcounter: {#2}, #1}
\newtcolorbox[auto counter, number within=section]{funcbox}[2][]{%
    title= {#2}, #1}

\newcommand{\erclogowrapped}[1]{%
\setlength\intextsep{0pt}%
\begin{wrapfigure}[3]{r}{#1}%
\includegraphics[width=#1]{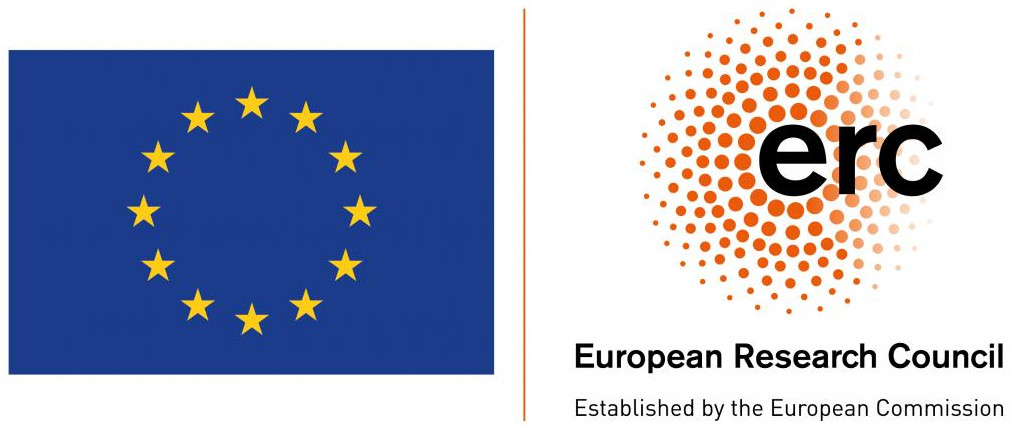}%
\end{wrapfigure}%
}

\title{Faster submodular maximization for several classes of matroids}

\author{Monika Henzinger\thanks{
\erclogowrapped{5\baselineskip} 
This 
project has received funding from the European Research Council (ERC)
under the European Union's 
Horizon 2020 research and innovation programme
(Grant agreement No.\ 101019564 ``The Design of Modern Fully Dynamic Data
Structures (MoDynStruct)'' and from the Austrian Science Fund (FWF) project
``Static and Dynamic Hierarchical Graph Decompositions'', 
I 5982-N, and project 
``Fast Algorithms for a Reactive Network Layer (ReactNet)'', P~33775-N, with
additional funding from the \textit{netidee SCIENCE Stiftung}, 2020--2024.
}
\\
Institute of Science and Technology Austria 
\and
Paul Liu
\\
Stanford University
\and 
Jan Vondr\'{a}k
\\
Stanford University
\and
Da Wei Zheng
\\
University of Illinois Urbana-Champaign
}
\date{}

\begin{document}

\maketitle
\thispagestyle{empty}

\begin{abstract}
The maximization of submodular functions have found widespread application in areas such as machine learning, combinatorial optimization, and economics, where
%In such applications, p
practitioners often wish to enforce various constraints; the matroid constraint has been investigated extensively due to its algorithmic properties and expressive power. Though tight approximation algorithms for general matroid constraints exist in theory, the running times of such algorithms typically scale quadratically, and are not practical for truly large scale settings.
Recent progress has focused on fast algorithms for important classes of matroids given in explicit form. Currently, nearly-linear time algorithms only exist for graphic and partition matroids~\cite{EN19}. In this work, we develop algorithms for monotone submodular maximization constrained by graphic, transversal matroids, or laminar matroids in time near-linear in the size of their representation. Our algorithms achieve an optimal approximation of $1-1/e-\eps$ and both generalize and accelerate the results of Ene and Nguyen~\cite{EN19}. In fact, the running time of our algorithm cannot be improved within the fast continuous greedy framework of Badanidiyuru and Vondrák~\cite{BV14}.

To achieve near-linear running time, we make use of dynamic data structures that maintain bases with approximate maximum cardinality and weight under certain element updates. These data structures need to support a weight decrease operation and a novel {\sc Freeze} operation that allows the algorithm to freeze elements (i.e. force to be contained) in its basis regardless of future data structure operations.
For the laminar matroid, we present a new dynamic data structure using the top tree interface of Alstrup, Holm, de Lichtenberg, and Thorup~\cite{AHLT05} that maintains the maximum weight basis under insertions and deletions of elements in $O(\log n)$ time. This data structure needs to support certain subtree query and path update operations that are performed every insertion and deletion that are non-trivial to handle in conjunction.
For the transversal matroid the {\sc Freeze} operation corresponds to requiring the data structure to keep a certain set $S$ of vertices matched, a property that we call $S$-stability. While there is a large body of work on dynamic matching algorithms, none are $S$-stable \emph{and} maintain an approximate maximum weight matching under vertex updates. We give the first such algorithm for bipartite graphs with total running time linear (up to log factors) in the number of edges.
\end{abstract}

\pagebreak

\clearpage
\pagenumbering{arabic} 

\section{Introduction}
Submodular optimization is encountered in a variety of applications -- combinatorial optimization, information retrieval, and machine learning, to name a few~\cite{B22}. Many such applications involve constraints, which are often in the form of cardinality or weight constraints on certain subsets of elements, or combinatorial constraints such as connectivity or matching. A convenient abstraction which has been studied heavily in this context is that of a {\em matroid constraint}\footnote{A matroid on $N$ is a family of ``independent sets'' $\cI \subset 2^N$ which is down-closed, and satisfies the extension axiom: For any $A,B \in \cI$, if $|A|<|B|$ then there is an element $e \in B \setminus A$ such that $A \cup \{e\} \in \cI$.}. For instance, transversal matroids appear in ad placement and matching applications~\cite{BK07,BHK08}, laminar and partition matroids capture capacity constraints on subsets which are widely used in recommendation settings (e.g. YouTube video recommendation algorithm~\cite{WRBJCG18}), and graphic matroids appear in applications for approximating Metric TSP~\cite{XR15}. Maximization of a submodular function subject to any of these constraints is an APX-hard problem, but a $(1-1/e)$-approximation is known in this setting for any matroid constraint \cite{CCPV11}, and the factor of $1-1/e$ is also known to be optimal \cite{NW78,F98}. Considering this, it has been a long-standing quest to develop fast algorithms for the submodular maximization problem that achieve an approximation close to the optimal factor of $1-1/e$. In this work, we achieve this goal for several common classes of matroids.

%For submodular maximization, the gold standard is a $(1-1/e)$-approximation \cite{CCPV11}, as better approximation ratios are known to be NP-hard~\cite{FNW78}. 
Perhaps the first step in this direction was the \emph{threshold-greedy technique} which gives a fast $(1/2-\epsilon)$-approximation \cite{BMKK14} for the cardinality constraint. With more work, this technique can be extended to give approximations close to $1-1/e$ \cite{BV14} and ultimately a $(1-1/e-\eps)$-approximation in running time $O(n/\eps)$ was found for the cardinality constraint.\footnote{Running time in this paper includes value queries to the objective function $f(S)$ as unit-time operations.}

 For general matroids, the fastest known algorithm is the \emph{fast continuous greedy algorithm}, which uses $O(nr \eps^{-4} \log^2 (n/\eps))$ oracle, where $n$ is the number of elements in the matroid and $r$ is the rank of the matroid \cite{BV14}. (The exact running time would depend on the implementation of these queries, which \cite{BV14} does not address.) We can assume that the rank $r$ scales polynomially with $n$ and hence this algorithm is not near-linear. Further work on fast submodular optimization developed in the direction of parallelized and distributed settings (see~\cite{MZLK19, LFKK22, LV19} and the references therein); we do not discuss these directions in this paper.

A recent line of work initiated by Ene and Nguyen~\cite{EN19} attempts  to develop a ``nearly-linear" continuous greedy algorithm, i.e., with a running time of $n \cdot \hbox{poly}(1/\eps, \log n)$. They achieved this goal for partition and graphic matroids. Prior to their work, the fastest known algorithm for any matroid class beyond cardinality was the work of Buchbinder, Feldman, and Schwartz~\cite{BFS17}, who showed an $O(n^{3/2})$-time algorithm for partition matroids.
%and gave nearly-linear time algorithms. 

This immediately leads to the question of whether such improvements are also possible for other classes of matroid constraints. As observed by Ene and Nguyen~\cite{EN19}, even determining feasibility may take longer than linear time for certain matroids. One such example is for matroids represented by linear systems, as simply checking the independence of a linear system takes $O(\rank(\cM)^\omega)$ time, where $\rank(\cM)$ is the rank of the linear system and $\omega$ is the exponent for fast matrix multiplication.

\subsection{Our contributions} 
\label{sec:contributions}
In this paper, we generalize and significantly improve the work of Ene and Nguyen \cite{EN19}, to develop a $(1-1/e-\eps)$-approximation for maximizing a monotone submodular function subject to a matroid constraint, for several important classes of matroids: namely for graphic, laminar, and transversal matroids. The technical developments behind these results are on two fronts: \\
(i) a refinement of the optimization framework of \cite{EN19} and formulation of abstract data structures required for this framework; \\
(ii) implementation of such data structures for graphic, laminar and transversal matroids. (See \Cref{sec:preliminaries} for definitions of these classes of matroids.)

To describe our results in more detail, 
the efficiency of optimizing a submodular function $f$ can be broken down into two components: the number of oracle calls to $f$, and the number of additional arithmetic operations needed to support the algorithm optimizing $f$. The number of oracle calls to $f$ that our framework need to achieve a $(1-1/e-\eps)$-approximation is $O_\eps(n\log^2(n))$ regardless of the matroid, where $n$ is the number of elements in the matroid constraint. Thus for all results below, the running time is measured in the number of the number of arithmetic operations needed by the data structures supporting the optimization of $f$. Our contributions are as follows:
\begin{itemize}
    \item We give nearly-linear time versions of continuous greedy for laminar, graphic, and transversal matroids. %Note that the laminar matroid is a generalization of the partition matroid considered by~\cite{EN19}. 
    These algorithms are accelerated by special data structures we developed for each matroid and which might be of independent interest. For all of our matroids, it is impossible to improve our running time without improving the continuous greedy algorithm itself.
    
    \item For graphic matroids on $n$ vertices and $m$ edges, we improve the running time of~\cite{EN19} from $O_\eps(n \log^5 n + m \log^2 n)$ to $O_\eps(m \log^2 n)$. 
    
    \item We generalize the partition matroids results of~\cite{EN19} to laminar matroids, and match their running time of $O_\eps(n \log^2 n)$ for continuous greedy. As a by-product, we also develop the first data structure that maintains the maximum weight basis on a laminar matroid with $O(\log n)$ update time for insertions and deletions that may be of independent interest. This data structure uses the top tree interface of 
    Alstrup, Holm, de Lichtenberg, and Thorup~\cite{AHLT05}.
    
    \item For transversal matroids represented by a bipartite graph with $m$ edges and the ground set being one side of the partition with $n$ vertices, we give an algorithm running in $O_\eps(m \log n + n\log^2 n)$ time.\footnote{Any transversal matroid with $n$ elements can be represented as the family of matchable sets the left-hand side of an $n \times n^2$ bipartite graph, with degrees at most $n$. See \Cref{sec:preliminaries} for more details.} This is the first such fast algorithm for transversal matroids. 
    
    For this we develop a  dynamic matching algorithm in a vertex-weighted, bipartite graph with a weighted vertex sets $L$ and an unweighted vertex set $R$ with the following conditions: (i) There exists a dynamically changing set $S  \subseteq L$ such that every vertex in $S$ must be matched in the current matching. (ii) The matching must give both an approximation in terms of cardinality in comparison to the maximum cardinality matching \emph{as well as} the weight of the matching in comparison to the best matching that matches all vertices of $S$, where the weight of the matching is the sum of the weights of the matched vertices of $L$. 
\end{itemize}
We emphasize that our results are true running times, as opposed to black-box independence queries to the matroid. The only black box operation we need is the query to the objective function $f(S)$. 

The performance of the dynamic matching data structure used in our transversal matroid algorithm is interesting as none of the earlier work on dynamic matching can maintain both a constant approximation to the weight as well as to the cardinality of the matching. Our algorithm builds on a recent fast algorithm for maximum-weight matching by Zheng and Henzinger~\cite{NEW}. We briefly mention a few relevant works:

\begin{itemize}
\item There is a conditional lower bound based on the OMv conjecture~\cite{HenzingerKNS15} that shows that  maintaining a \emph{maximum weight matching} in an edge weighted bipartite graph with only \emph{edge weight increase operations} cannot be done in amortized time $O(n^{1-\delta})$ per edge weight increase and amortized time $O(n^{2-\delta})$ per query operation for any $\delta > 0$~\cite{HenzingerP021}. %This lower bound is based on the corresponding conditional lower bound of Dahlgaard for incremental matching~\cite{Dahlgaard16}. As the latter paper also contains a conditional lower bound for decremental algorithms, 
The reduction from~\cite{HenzingerP021} can be adapted to the setting with only \emph{edge weight decrement} operations, achieving the same lower bound. Thus, this shows that our running time bound cannot achieved if a \emph{exact} maximum weight matching has to be maintained under edge weight decrement operations.

\item Le, Milenkovic, Solomon, and Vassilevska-Williams~\cite{LeMSW22} studied one-sided vertex updates (insertions and deletions) in bipartite graphs and gave a \emph{maximal matching} algorithm whose total running time is $O(K)$, where $K$ is the total number of inserted and deleted \emph{edges}. Bosek el al.~\cite{BLSZ14}  studied one-sided vertex updates (either insertions-only or deletions-only) in bipartite graphs  and gave algorithms that  for any $\eps > 0$ maintain a $(1-\eps)$-approximate maximum cardinality matching in total time $O(m/\eps)$. Gupta and Peng~\cite{GuptaP13} developed the best known dynamic algorithm that allows both \emph{edge} insertions and deletions and maintains a $(1-\eps)$-approximate maximum cardinality matching (for any $\eps > 0$). It requires time
$O(\sqrt m/\eps)$ amortized time per operation. For all these algorithms, they either cannot be extended to the weighted setting, or cannot maintain both a constant approximation to the weight as well as to the cardinality of the matching.
\end{itemize}

\subsection{Technical Overview}

\paragraph{The submodular optimization framework.}
Our framework is an adaptation and improvement over that of Ene and Nguyen~\cite{EN19}. The framework consists of two phases:
\begin{enumerate}[label={(\arabic*)}]
    \item a {\sc LazySamplingGreedy} phase, which aims to build a partial solution that either provides a good approximation on its own, or reduces the problem to a residual instance with bounded marginal values of elements.
    \item a {\sc ContinuousGreedy} phase, which is essentially the original fast continuous greedy algorithm \cite{CCPV11,BV14}, with an improved analysis based on the fact the marginal values are bounded.
\end{enumerate}
The original fast {\sc ContinuousGreedy} runs in time $O_\eps(nr\log^2 n)$, where the factor of $r\log^2 n$ is due to the cost in evaluating the multilinear extension of $f$. The multilinear extension is an average over values of $f$ on randomly drawn subsets of the input. In {\sc ContinuousGreedy}, $O(r\log^2 n)$ samples are needed to estimate the multilinear extension to sufficient accuracy.

The {\sc LazySamplingGreedy} phase transforms $f$ into a function $\tilde{f}$ such that (a) the number of samples required in {\sc ContinuousGreedy}
for $\tilde{f}$ is reduced by a factor of $r$ and (b) the optimal solution of $\tilde f$ is within a $(1-\eps)$ factor of the optimal solution $f$. {\sc LazySamplingGreedy} does this by constructing an initial independent set $S$ in our matroid $\cM$ such that $\tilde{f}(T) = f(T\cup S)$ has relatively small marginal values. The final solution is obtained by running {\sc ContinuousGreedy} on $\tilde{f}$ constrained by $\cM \setminus S$ and combining the solution with $S$. The construction of $S$ relies on a fast data structure to get the maximum weight independent set of $\cM$ at any given time, subject to weight changes on each element. 

We begin by simplifying and improving the {\sc LazySamplingGreedy} phase (\Cref{sec:lazysample}). A significant part of {\sc LazySamplingGreedy} in \cite{EN19} is dedicated to randomly checking and refreshing estimates of marginal values for each element in the matroid. We show that (a) this random checking can be dramatically reduced, and (b) the maximum-weight independent set requirement for the data structure can be relaxed to a constant-factor approximation of the maximum-weight independent set. The relaxation to constant-factor approximations enables us to design much more efficient data structures than the previous work of \cite{EN19}, and also allows us to handle more classes of matroids.

In the {\sc ContinuousGreedy} phase (\Cref{sec:continuousgreedy}), a subroutine requires an independence oracle to check if proposed candidate solutions are independent sets of the matroid. Ene and Nguyen use fully dynamic independence oracles to implement this, where independence queries require $O(\textrm{polylog} \ n)$ time. We show that only incremental independence oracles are needed. This opens the door to implementing such oracles for new classes of matroids, as sublinear fully dynamic independence oracles are provably hard in many settings (e.g. bipartite maximum-cardinality matchings are hard to maintain exactly in faster than $O(n^{2-\eps})$ amortized time per update \cite{AbboudW14}, which corresponds to finding the maximum independent set in a transversal matroid).

\paragraph{Dynamic data structures for various matroid constraints.}
From a data structures point of view, we give the first efficient data structure for two settings. (We refer the reader to \Cref{sec:preliminaries} for definitions of laminar, graphic and transversal matroids.)

\emph{Maximum-weight basis in a laminar matroid.} In a laminar matroid with $n$ elements we are able to output the maximum-weight basis under an online sequence of insertions and deletions of weighted elements with $O(\log n)$ time per update, which we show to be worst-case optimal. 
The biggest challenge for laminar matroids is that each element may have as many as $O(n)$ constraints that need to be kept track of on each insertion and deletion. We leverage the tree structure induced by a laminar set system to build data structures.
%Even then, the problem is not easy, as on an element insertion we need to handle subtree operations like getting the smallest element in the basis, on element deletions to figure out the largest element that can be added to the basis, and whenever the basis changes we need to support path operations to update many constraints on a path to the root.
%
Specifically, we show there exists a data structure with $O(\log^2 n)$ query time using the heavy-light decomposition technique of \cite{Tarjan75}.
Since the heavy-light decomposition technique decomposes the tree into paths, we store some carefully chosen auxiliary information at each vertex to transform our subtree queries into path queries.
We further improve this to $O(\log n)$ using the top tree interface of 
Alstrup, Holm, de Lichtenberg, and Thorup~\cite{AHLT05}
that more naturally support both path and subtree operations using a similar idea of carefully storing the right auxiliary information at each top tree node combined with lazily propagating path changes.

\emph{Approximate maximum-weight basis in a transversal matroid.}
In the case of transversal matroids, our {\sc LazySamplingGreedy+} algorithm requires what we call a \emph{$(c, d)$-approximate maximum weight matching oracle}: a matching that is an $c$-approximation in weight and at least $d$ of the cardinality of the maximum cardinality matching. In addition, the oracle must implement two update operations, (a) a {\sc Freeze}  operation that adds a vertex of $L$ to $S$ and (b) a {\sc Decrement} operation that reduces the weight of a vertex of $L$ to a given value. We give a novel algorithm that maintains a maximal (inclusion-wise) matching $M$ in a weighted graph such that the weight of the matching is a $(1-\epsilon)$-approximation of the max-weight solution in total time $O(m (1/\epsilon + \log n))$. Thus our algorithm is a $(1-\epsilon, 1/2)$-approximate maximum weight matching oracle. Due to standard rescaling techniques we can assume that the maximum weight $W = O(n)$. 

To illustrate the challenges introduced by {\sc Freeze} operations, assume we want to maintain a $(1/2,1/2)$-approximate maximum weight matching oracle and let $n\ge 5$ be an odd integer. Consider a graph consisting of the path $\ell_0, r_0, \ell_1, r_2, \dots , \ell_{(n+1)/2}$ of length $n-1$ where the first and last edge have large weight $W$< say $W = 100 n$, and all other edges have weight 1. If the initial $(\Omega(1), \Omega(1))$-approximate matching has greedily picked every second edge, it achieves a weight of $W + \frac{n-1}{2} -1$ versus an optimum weight of $2W + \frac{n-1}{2} - 2.$ 
Now if the weight of the first edge is halved, the weight of the computed solution drops to $W/2 + \frac{n-1}{2} -1$, which is no longer a 2-approximation of the weight of the optimum solution (for large enough $W$). Thus the algorithm needs to match the last edge in order to maintain a 2-approximation to the weight. This would only require two changes to the matching, namely un-matching the edge $(\ell_{(n-1)/2}, r_{(n-1)/2})$
and matching the last edge. However, if prior {\sc Freeze} operations added the vertices $\ell_1, \ell_2, ... \ell_{(n-1)/2}$ to
$S$, i.e. $S = V \setminus \{ \ell_0\}$, then we cannot un-match $\ell_{(n-1)/2}$. Thus, the edge $(\ell_{(n-1)/2}, r_{(n-3)/2})$ needs to be unmatched, which in turn un-matches the
vertex $\ell_{(n-3)/2}$, leading to further un-matchings and matchings along the path. More specifically, {all} $(n-1)/2$ matched edges need to change. Next, the weight of the last edge is divided by 4 and the process along the path starts again.
Thus, due to the prior {\sc Freeze} operations, each {\sc Decrement} operation can lead to $\Theta(n) = \Theta(m)$ changes in the graph.
As this can be repeated $\Theta(\log W) = \Theta(\log n)$ times it shows that work $\Omega(m \log n)$ is unavoidable if a 2-approximation to the weight is maintained with {\sc Freeze} operations. This is the running time that we achieve.

More precisely, for any $\eps > 0$, we give an $(1-\epsilon, 1/2)$-approximate maximum weight matching oracle under {\sc Freeze} and {\sc Decrement} operations with total time $O(m (\log n + 1/\eps))$.
It follows that $O(m (\log n + 1/\eps))$ is also an upper bound on the number of matching and un-matching operations of the algorithm. To do so we extend
a recent algorithm~\cite{NEW} for fast $(1-\epsilon)$-approximate maximum weight matching algorithm in bipartite graphs based on the multiplicative weight update idea. The analysis within~\cite{NEW} shows stability properties that makes the {\sc Freeze} operation trivial to implement. However, {\sc LazySamplingGreedy+} requires that the maintained matching is at least $1/2$ the size of the maximum matching. Furthermore, we need to support the {\sc Decrement} operation. We show that both extensions are possible and obtain an algorithm with total time (for all operations) of $O(m (1/\epsilon + \log n))$.

\section{Preliminaries}
\label{sec:preliminaries}
\paragraph{Set notation shorthands.} Given a set $S \subseteq \cN$ and an element $u \in \cN$, we denote $S \cup \{u\}$ and $S \setminus \{u\}$ by $S + u$ and $S - u$ respectively. Similarly, given sets $S, T \subseteq \cN$, we denote $S \cup T$ and $S \setminus T$ by $S+T$ and $S-T$ respectively.

\paragraph{Submodular functions.}
Given a set $\cN$, a set function $f\colon 2^\cN \to \mathbb{R}$ is called \emph{submodular} if for any two sets $S$ and $T$ 
we have
$$
	f(S) + f(T) \geq f(S + T) + f(S - T).
$$
We only consider monotone submodular functions, where $f(S)\leq f(T)$ for any sets $S\subseteq T$.

\paragraph{Matroids.}
A set system is a pair $\cM = (\cN, \cI)$, where $\cI \subseteq 2^{\cN}$. We say that a set $S \subseteq \cN$ is \emph{independent} in $\cM$ if $S \in \cI$. The \emph{rank} of the set system $M$ is defined as the maximum size of an independent set in it. The independent sets must satisfy (i) $S\subseteq T, T \in \cI \implies S \in \cT$, and (ii) $S, T \in I, |S| < |T| \implies \exists e \in T \setminus S$ such that $S + e \in \cI$.

A matroid constraint means that $f$ is optimized over the independent sets of a matroid.

\paragraph{Graphic matroids.}
Let $G = (V,E)$ be a graph. A graphic matroid has $\cN = E$ and $\cI$ equal to the forests of $G$. The rank of $\cM$ is $|V|-C$ where $C$ is the number of connected components in $G$.

\paragraph{Laminar matroids.}
Let $\{S_1, S_2, \ldots, S_m\}$ be a collection of subsets of a set $\cN$ such that for any two intersecting sets $S_i$ and $S_j$ where $i \neq j$, either $S_i \subset S_j$ or $S_j \subset S_i$. Furthermore, let there be non-negative integers $\{c_1, c_2, \ldots, c_m\}$ associated with the $S_i$'s. Let $\cI$ be the sets $S \subseteq \cN$ for which $|S\cap S_i| \leq c_i$ for all $i$. $\cI$ is then the collection of independent sets in a laminar matroid on $\cN$. A laminar matroid has a natural representation as a tree on $m$ nodes, which we describe in \Cref{sec:laminar}.

\paragraph{Transversal matroids.}
Let $G=((L,R), E)$ be a bipartite graph with a bipartition of vertices $(L, R)$. Let $\cI$ be the collection of sets $S\subseteq L$ such that there is a matching of all vertices in $S$ to $|S|$ vertices in $R$. A transversal matroid has $\cN = L$ and $\cI$ as its independent sets. From the definition, it is clear that $\rank(\cM)$ is the size of the maximum cardinality matching in $G$. 
It is known that every transveral matroid can be represented by a bipartite graph where $L$ is the ground set of the matroid and $|R| = \rank(\cM)$ (see \cite{Schrijver-book}, Volume B, equation (39.18)).

%%%%%%%%%%%%%%% FRAMEWORK %%%%%%%%%%%%%%%%%%%

\section{Improved nearly-linear submodular maximization}
In what follows, $f$ is the submodular function we want to maximize, $\cM$ is the matroid constraint, $n$ is the number of elements in $\cM$, and $OPT$ is the optimal independent set. Additionally, we can assume that $\rank(\cM) = \omega(\log n)$, as the standard {\sc ContinuousGreedy} would run in $O(n\,\mathrm{polylog} n)$ time otherwise. 

Our high-level framework adapts and improves upon the nearly-linear time framework of Ene and Nguyen~\cite{EN19}. We will do the following:
\begin{algbox}[label=alg:framework]{Overall Framework}
\begin{enumerate}
    \item Run {\sc LazySamplingGreedy+} (see below), to obtain a partial solution $S_0$.
    \item Run {\sc ContinuousGreedy} on $\tilde{f}(T) = f(S_0 \cup T) - f(S_0)$ with the constraint $\cM / S_0$, to obtain a solution $S_1$.
    \item Return $S_0 \cup S_1$.
\end{enumerate}
\end{algbox}

As previously discussed, the original {\sc ContinuousGreedy} runs in time $O_{\eps}(nr\log^2 n)$, where the $r\log^2 n$ is due to the number of samples needed to evaluate the multilinear extension of $f$. {\sc LazySamplingGreedy+} finds a set $S_0$ such that $\tilde{f}(T) := f(T \mid S_0) = f(S_0 \cup T) - f(S_0)$ has a tighter range of marginal values. This allows us to reduce the number of samples used in {\sc ContinuousGreedy} by a factor of $r$.

The overall idea is to run {\sc LazySamplingGreedy+} until the marginals in $\tilde{f}$ are small enough to guarantee good performance in the {\sc ContinuousGreedy} phase of our overall framework (Algorithm~\ref{alg:framework}). To accelerate our algorithms, we construct specialized data structures for both {\sc LazySamplingGreedy+} and {\sc ContinuousGreedy}. 
%We begin with the description of these data structures.

\subsection{Data structure requirements}
\label{sec:ds-reqs}
We next describe the data structures needed for the two phases of our algorithm. In the {\sc LazySamplingGreedy+} phase we need a $c$-approximate dynamic max-weight independent set oracle. In the {\sc CountinuousGreedy} phase we have two options of dynamic independence oracles, both of them unweighted. In addition, our {\sc LazySamplingGreedy+} requires the ability to obtain a weighted sample from the approximate max-weight independent set. Since we use these data structures as subroutines in our static algorithm which uses the answers of the data structure to determine future updates, it is important that their running time bounds are valid against an \emph{adaptive} adversary.

\paragraph{Dynamic $(c,d)$-approximate maximum weight oracle.}
Let $\cM=(\uni, \cI)$ be a matroid. Given an independent set $S \subseteq \uni$
the \emph{independent sets relative to $S$} are the independent sets of $\cM$ that contain $S$.
Let $rank(\cM)$ denote the size of the largest independent set, which equals the size of the largest independent set containing $S$. The weight of an independent set is the sum of the weights of its elements.
A \emph{maximum weight basis in $\cM$ relative to $S$} is a basis $B^*$ that maximizes the sum of $w_e$ over all bases of $\cM$ that contain $S$.

Let $c<1$ and $d<1$ be constants.
An independent set $B$ is called an \emph{$(c,d)$-approximate independent set relative to $S$} if it fulfills the following conditions: (a) its size is at least $\rank(\cM)\cdot d$
and
(b) its weight is at least a $c$-approximation to the weight of a maximum weight basis relative to $S$.

We study the dynamic setting where each element $e \in \uni$ has a dynamically changing weight $w_e \in \RR^+$ and where $S$ is a dynamically growing subset of $\uni$. 
A \emph{$(c,d)$-approximate dynamic maximum weight oracle} is a data structure which maintains a $(c,d)$-approximate independent set $B$ relative to $S$ (i.e. in the matroid $\cM / S$) 
while $S$ and the weight of elements not in $S$ can change.
Initially $S$ is an empty set and the data structure supports the following operations: 
\begin{itemize}
    \item {\sc Freeze($e$):} Add to $S$ the element $e$, where $e$ must belong to the current $(c,d)$-approximate basis relative to (the old) $S$  and return the changes to $B$.
   
    \item {\sc Decrement($e,w$):} Return the weight $w_e$ of $e \notin S$ to $w$, which is guaranteed to be smaller than the current weight of $e$  and return the changes to $B$.

    \item {\sc ApproxBaseWeight():} Return the weight of the $(c,d)$-approximate independent set maintained by this data structure.
\end{itemize}

If $c=1$ and $d=1$ we call such a data structure a \emph{dynamic maximum weight oracle relative to $S$.}

We will use $(c,d)$-approximate maximum weight oracles in the {\sc LazySamplingGreedy+} phase of the algorithm.  

We will also need to augment this data structure with two additional sampling operations.
Whenever the independent set $B$ maintained by the data structure changes, we need to spend an extra $O(1)$ time updating a sampling data structure.
This sampling data structure can be generically and efficiently implemented to augment any $(c,d)$-approximate maximum weight oracle as long as the $(c,d)$-approximate maximum weight oracle does not change the independent set $B$ too much amortized over all calls to the data structure.
This is described in \Cref{sec:sampling_ds}.

\begin{itemize}
    \item {\sc Sample($t$):} Return a subset of $B\setminus S$, where each element is included independently with probability $\min\left(1, \frac{t}{w(B\setminus S)}w_e\right)$.
    \item {\sc UniformSample():} Return a uniformly random element from $B\setminus S$.
\end{itemize}

\paragraph{$(1-\epsilon)$-approximate independence oracles.}
For the second phase of our algorithm {\sc ContinuousGreedy} we have a choice between two data structures. Both of them are unweighted, i.e., elements have no associated weights.
We can either use an incremental (i.e. insertions-only) exact data structure or a dynamic $(1-\epsilon)$-approximate data structure, for a small $\epsilon >0$.
Next we define both in more details.
 
\emph{Incremental independence oracle.}
The incremental independence oracle data structure maintains an independent set $B$ and supports the following operation:
\begin{itemize}
    \item {\sc Test($e$):} Given an element $e$, decide if $B \cup \{e\}$ is independent. If so, output YES, otherwise output NO.
    \item {\sc Insert($e$):} Given an element $e$ such that $B \cup \{e\}$ is independent, add $e$ to $B$.
\end{itemize}

\emph{$(1-\epsilon)$-approximate dynamic maximum independent set data structure.}
Let $\epsilon >0$ be a  small constant and
let us call an independent set $B$ of a matroid $\cM$ that contains at least $(1-\epsilon) \cdot \rank(\cM)$ elements an  \emph{$(1-\epsilon)$-approximate basis} of $\cM$.
The $(1-\epsilon)$-approximate data structure
maintains an $(1-\epsilon)$-approximate basis $B$ for a dynamically changing matroid $\cM$ and supports the following operations.
\begin{itemize}
    \item {\sc Batch-Insert($E'$):} Given a set $E'$ of new elements, insert all elements of $E'$ into the matroid  $\cM$ and compute a new $(1-\epsilon)$-approximate basis $B$ such that all elements that were in the basis before the update belong to $B$. Return all new elements that were introduced to $B$.
    \item {\sc Delete($e$):} Given an element $e$ of $\cM$, delete $e$ from $\cM$ and update the independent set $B$ such that it consists of at least $(1-\epsilon) \cdot \rank(\cM)$ elements of the new $\cM$. If any new elements were added to $B$, return this set of new elements. Otherwise, return $\emptyset$. 
\end{itemize}

Depending on which version of the algorithm we use, we will need either an exact incremental oracle or a $(1-\epsilon)$-approximate dynamic maximum independent set data structure.

\subsection{The {\sc LazySamplingGreedy+} algorithm}
\label{sec:lazysample}
In this section, we describe the implementation of {\sc LazySamplingGreedy+}.

The {\sc LazySamplingGreedy+} algorithm is inspired by the Random Greedy algorithm of Buchbinder, Feldman, Naor, and Schwartz~\cite{BFNS14} and the Lazy Sampling Greedy algorithm of Ene and Nguyen~\cite{EN19}. The algorithm begins with an initially empty solution $S$, and runs until the function $\tilde{f}(T) = f(T|S)$ has small enough marginals to reduce the sampling requirements of {\sc ContinuousGreedy}. 

We denote the weight of an element by $w_e(S) := f(S\cup\{e\})-f(S)$, and $\textrm{weight}(T)$ to denote $\sum_{e \in T}w_e(S)$. The algorithm will only ever add elements to $S$, so by submodularity, $w_e(S)$ can only decrease as the algorithm runs (satisfying the requirements of \Cref{sec:ds-reqs}). Throughout this algorithm, we use a $(c,d)$-approximate maximum-weight oracle (\Cref{sec:ds-reqs}) that maintains a maximum-weight independent set $B$ as the weights $w_e(S)$ are updated. For the sake of exposition, we assume $c\geq 1/2$, and $d \geq 1/2$.

\paragraph{Discretizing the marginal weights.}
Whenever $S$ is changed, the weight $w_e(S)$ of all elements $e$ can be changed. To reduce the number of weight changes, we use a standard rounding trick. Assume we have some constant-factor approximation $M$ to $f(OPT)$ (which can be computed in $O(n)$ time via well-known algorithms~\cite{BFS17}). Instead of maintaining $w_e(S)$ exactly, we round $w_e(S)$ to one of logarithmically many weight classes, that is, $w_e(S)$ belongs to weight class $j$ if $w_e(S)\in ((1-\eps)^{j+1} M, (1-\eps)^j M]$, with the lowest class containing all weights from $[0, O(\eps M/r)]$. The value of the rounded weight is then $\tilde{w}_e = (1-\eps)^{j_e}$. We denote by \emph{bucket} ${\cal B}^{(j)}$ all elements that belong to weight class $j$. Throughout the algorithm, we maintain estimates $\tilde{j_e}$ for the weight class that $e$ is in (and thus estimates of $w_e$ as well). An estimate is called \emph{stale} if $w_e(S)$ is not actually in the weight class indicated by $\tilde{j_e}$. To achieve a multiplicative error of $(1-\eps)$, it suffices for the number of different weight classes to be at most $O(\eps^{-1} \log (r/\eps))$, where $r$ is the rank of the matroid. We denote by $\textrm{weight}(B)$ the sum of current weight estimates over the set $B$.

The analysis of our algorithm works with any constant-factor approximation to $f(OPT)$ and any constant $c$-approximate maximum weight independent set data structure, albeit with slight changes in the approximation factors.

\

\begin{algbox}{\sc LazySamplingGreedy+}
$S \leftarrow \emptyset$, and set the weight estimate $\tilde{w}_e$ to $w_e(\emptyset)$ for all $e \in \cM$. \\
$\cD \leftarrow \text{$(c,d)$-approximate dynamic maximum weight oracle on $\cM$ and $\tilde{w}$.}$ \\
\vspace{0em} \\
While $\cD.\mbox{\sc ApproxBaseWeight}() \ge \frac{50}{\eps} f(OPT)$:
\begin{enumerate}
\item $B' \gets \cD.\mbox{\sc Sample}(128 \log n)$  \\
(a random subset of $B \setminus S$, each element included independently with probability
 $p_e = \min \{ 1, \frac{128 \log n}{\tilde{w}(B \setminus S)} \tilde{w}_e \}$)
\label{algline:random-sample}
\item Update the weights of all stale elements $e \in B'$ by computing $j_e$, $\tilde{w}_e = (1-\eps)^{j_e}$ and then calling $\cD$.\mbox{\sc Decrement}$(e, \tilde{w}_e)$. \label{algline:weight-update}
\item If less than half of the elements in $B'$ where $p_e = 1$ were stale (i.e. needed an update), and less than half of the elements in $B'$ where $p_e<1$ were stale, then add $e = \cD.\textsc{UniformSample}()$ to $S$ by calling $\cD$.\mbox{\sc Freeze}$(e)$. \label{algline:add-to-B}
\end{enumerate}
\end{algbox}

Note that in each iteration, we check and update only the weights of some random sample of elements.
This is for efficiency; we show the estimated weight $\tilde{w}(B)$ is still correct in expectation up to a constant multiplicative factor. We begin the correctness proof by showing the following lemma.
\begin{lemma}
\label{lem:at-least-gain}
Assume $0 < \eps < 1/3$.
With high probability, if less than $\frac12$ of the estimated weight of $B'$ is in elements which are stale, then 
$$ \sum_{e \in B \setminus S} w_e(S) \geq \frac{4}{\eps} f(OPT).$$
\end{lemma}

\begin{proof}
First, observe that $\sum_{e\in S} \tilde{w}_e \leq (1+\eps) f(S)$. This is because each element in $S$ has their weight $\tilde{w}_e$ frozen when they enter $S$, and at that point $\tilde{w}_e$ is at most a $(1+\eps)$ factor off from their true marginal value on top of $S$.  Also, since $S$ is a feasible solution, we have $f(S) \leq f(OPT)$.
As long as the while loop is running, we have $\sum_{j \in B} \tilde{w}_e \ge \frac{50}{\eps} f(OPT)$,
and hence 
\begin{equation*}
\sum_{e \in B\setminus S} \tilde{w}_e \geq \sum_{e \in B} \tilde{w}_e - (1+\eps)f(S) \geq \frac{48}{\eps} f(OPT). \end{equation*}

Next, we observe that the expected number of elements included in $B'$ is $\sum_{e \in B \setminus S} p_e \leq 128 \log n$. Hence by the Chernoff bound, 
$$ \Pr[|B'| > 160 \log n] < e^{-(1/4)^2 128 \log n / 3} < \frac{1}{n^{2.5}} $$
and with high probability we have $O(\log n)$ elements to check.

Let $\tilde{B}$ be the subset of $B \setminus S$ which is not stale.
We distinguish two types of elements in $B \setminus S$: large elements, for which $\tilde{w}_e > \frac{\tilde{w}(B \setminus S)}{128 \log n}$, and small elements, for which $\tilde{w}_e \leq \frac{\tilde{w}(B \setminus S)}{128 \log n}$. At least half of the weight of $B \setminus S$ must come from either the large elements or the small elements. In the first case, we check all the large elements, and if more than half of them are stale, we skip line 3. Hence the only way we can execute line 3 is that the non-stale large elements contain weight at least 
$\frac14 \sum_{e \in B \setminus S} \tilde{w}_e$.

In the second case, let $B_s$ denote the small elements in $B \setminus S$; we have $\sum_{e \in B_s} \tilde{w}_e \geq \frac12 \sum_{e \in B \setminus S} \tilde{w}_e$, and each small element is sampled with probability $p_e = \frac{128 \log n}{\tilde{w}(B \setminus S)} \tilde{w}_e$. Therefore, the expected number of small elements chosen in $B'$ is $\E[|B' \cap B_s|] \geq 64 \log n$, and by the Chernoff bound,
$$ \Pr[|B' \cap B_s| < 48 \log n] < e^{(1/4)^2 64 \log n / 2} = \frac{1}{n^2}.$$
On the other hand, if the weight of non-stale small elements is less than $\frac18 \sum_{e \in B \setminus S} \tilde{w}_e$, the expected number of such elements chosen in $B'$ is  $\E[|B' \cap B_s \cap \tilde{B}|] \leq 16 \log n$, and again by the Chernoff bound, the probability that more than this quantity is more than $24 \log n$ is less than $e^{-(1/2)^2 16 \log n / 2} = 1 / n^2$.

Hence either way, if the weight of non-stale elements in $B \setminus S$ is less than $\frac18 \sum_{e \in B \setminus S} \tilde{w}_e$, with high probability we do not execute line 3. Hence we can assume that  $\sum_{e \in \tilde{B}} w_e(S) \geq (1-\eps) \sum_{e \in \tilde{B}} \tilde{w}_e \geq \frac{1-\eps}{8} \sum_{e \in B \setminus S} \tilde{w}_e$ whenever we execute line 3.
Then, 
$$\sum_{e \in B \setminus S} w_e(S) \geq  \sum_{e \in \tilde{B}} w_e(S) \geq \frac{1-\eps}{8} \sum_{e \in B \setminus S} \tilde{w}_e \geq \frac{1-\eps}{8} \cdot \frac{48}{\eps} f(OPT) \geq \frac{4}{\eps} f(OPT).$$ The second inequality is due to the fact that $\sum_{e \in B \setminus S} \tilde{w}_e \geq \frac{48}{\eps} OPT$, and the fact that at least $1/8$ of the weight in $B \setminus S$ is not stale. The last inequality is due to the assumption that $\eps < 1/3$.
\qedhere
\end{proof}

Next, we show a bound on the computational complexity of {\sc LazySamplingGreedy+}.
\begin{lemma}
{\sc LazySamplingGreedy+} uses at most $O(n\eps^{-1} \log(r/\eps))$ arithmetic operations, calls to $f$, and calls to the maximum weight data structure. \label{lem:lg+-oracle-calls}
\end{lemma}
\begin{proof}
Recall that throughout the algorithm set $S$ only increases and, thus, the weight of each element only decreases.
Thus, by discretizing the weights, the weight of an element can only be changed at most $O(\eps^{-1}\log r)$ times and the total number of weight changes is at most $O(n\eps^{-1}\log (r/\eps))$. When all weights are at the lowest class, no element in $B$ contributes more than $O(\eps f(OPT) / r)$. With enough weight classes (e.g. more than $10\log_{1-\eps} (\eps/r)$ classes), the constant on the big-$O$ is less than 1. Thus $\sum_{e \in B} \tilde{w}_e < \eps f(OPT) < f(OPT) / \eps$, so the while loop must terminate.

The most costly part of the algorithm is line 2. Each time line 2 executes, $O(\log n)$ weight changes and calls to $f$ are made. Thus the overall cost of line 2 throughout the entire algorithm is at most $O(n\eps^{-1}\log (r/\eps))$.
\end{proof}

Next we observe that $S$ cannot have too many elements, otherwise $f(S)$ is close to $f(OPT)$ and we are done.

\begin{observation}
\label{obs:bound-on-S}
With high probability, $S$ at the end of the algorithm has at most $\eps r/2$ elements.
\end{observation}

\begin{proof}
Observe that whenever we add an element to $S$, it is a random element from a set $B \setminus S$ satisfying
$\sum_{e \in B \setminus S} \geq \frac{4}{\eps} f(OPT)$ with high probability (\Cref{lem:at-least-gain}). Hence, whenever we add an element to $S$, $f(S)$ increases by at least $(4-o(1)) f(OPT)/(\eps r)$ in expectation, and this is true conditioned on any prior history. Hence, by martingale concentration, the value after including $t$ elements is with high probability at least $2 t \cdot f(OPT) / (\eps r)$. If there was non-negligible probability that the algorithm includes more than $\eps r / 2$ elements, then with some probability the value of $f(S)$ in theses cases would exceed $f(OPT)$, which is impossible (as $S$ is always a feasible solution). 
Hence, with high probability, the algorithm does not include more than $\eps r / 2$ elements.
\end{proof}

\begin{theorem}
Let $S$ be the set returned at the end of {\sc LazySamplingGreedy+}, $OPT :=  \argmax_{T \in \cM} f(T)$, and $OPT^* := \argmax_{T \in \cM/S} f(T|S)$. The following inequality holds:
$$ \E[f(OPT^* \cup S)] \geq (1-2\eps) f(OPT).$$
\end{theorem}

\begin{proof}
By Observation~\ref{obs:bound-on-S}, we can assume that the algorithm includes at most $\eps r / 2$ elements in $S$.
Let $r=\rank(\cM)$, and $\{s_1, s_2, \ldots, s_t\}$ be the sequence of elements we add to $S$. 
Let $OPT$ be the basis maximizing $f$ and order $OPT=\{o_1, o_2, \ldots, o_r\}$ such that $\{s_1, \ldots, s_i\} \cup \{o_{i+1}, \ldots, o_r\}$ is independent for all $i$. We can arrange this ordering so that in each step, $o_{i}$ is a uniformly random one of the remaining element of $OPT$:  Given a base $\{s_1,\ldots,s_{i-1},o_i,\ldots,o_r\}$, $S_{i-1} = \{s_1,\ldots,s_{i-1}\}$, and another base $B$ containing $S_{i-1}$, there is a matching between $B \setminus S_{i-1}$ and $\{o_i,\ldots,o_n\}$ such that for any $e \in B \setminus S_{i-1}$, it is possible to add $e$ to $S_{i-1}$ and remove its matching optimal element (assume $o_i$), so that $\{s_1,\ldots,s_i,o_{i+1}, \ldots,o_n\}$ is still a base. Given that we choose $e \in B \setminus S_{i-1}$ uniformly at random, the choice of an optimal element to remove is also uniformly random.

To simplify the notation, we define $S_i := \{s_1, \ldots, s_i\}$ and $O_i := \{o_{i}, \ldots, o_r\}$, with the convention that $S_0 = O_{r+1} = \emptyset$ and use $f(s_i|S_{i-1})$  (resp. $f(o_i|O_{i-1})$) to denote
$f(S_i) - f(S_{i-1})$ (resp.~$f(O_i) - f(O_{i-1})$).
Note that it follows that $f(s_i|S_{i-1}) = w_{s_i}(S_{i-1})$.

Following ~\cite{EN19}, we will show that
\begin{equation}
\label{eq:en-lemma}
\E[f(s_{i} | S_{i-1})] \geq \frac{1}{2\eps} \E[f(o_{i} | O_{i+1})].
\end{equation}
Adding the inequality \Cref{eq:en-lemma} for $i=1,\ldots,t$, we have $$\E[f(S_t) - f(\emptyset)] = E[f(S)] \geq \frac{1}{2\eps} \E[f(O_1) - f(O_t)] = \frac{1}{2\eps} \E[f(OPT) - f(O_t)].$$ 

Since $OPT^* \cup S$ and $O_t\cup S$ are both independent, we would then have 
\begin{align*}
\E[f(OPT^* \cup S)] & \geq \E[f(O_t \cup S)] \\
              & \geq \E[f(OPT)] - 2\eps \E[f(S)] \\
              & \geq (1-2\eps) \E[f(OPT)].
\end{align*}

We now show \Cref{eq:en-lemma}. Fix an iteration $i$ and all random choices $s_1, s_2, \ldots, s_{i-1}$ up to iteration $i$.
By \Cref{lem:at-least-gain}, w.h.p. $\sum_{e \in B \setminus S_{i-1}} w_e(S_{i-1}) \geq \frac{4}{\eps} f(OPT)$ on line 3 for all $i$. Thus choosing a random $e \in B \setminus S_{i-1}$ as $s_i$ yields
\begin{equation}
\label{eq:en-lower}
\E[f(s_i | S_{i-1})] \geq \frac{f(OPT)}{\eps |B \setminus S_{i-1}|} \geq \frac{2f(OPT)}{\eps(r-i)}.
\end{equation}
As discussed above, since $s_i$ is chosen uniformly at random from $B \setminus S_{i-1}$, one of the $|B \setminus S_{i-1}|$ elements in $O_i$ is chosen as $o_i$. By our assumptions on $c$ and $d$ for the $(c,d)$-approximate maximum weight oracle, $|B| \geq r/2$. Furthermore, $i \leq \eps r$ due to \Cref{obs:bound-on-S}. Thus $|B\setminus S_{i-1}| \geq r/2 - i \geq (r-i)/4$ for $\eps < 1/3$. This means that any element in $O_i$ has probability at most $4/(r-i)$ chance of being chosen as $o_i$. Thus 
\begin{align}
\E[f(o_{i} | O_{i+1})] & \leq \sum_{o \in O_i} \frac{4}{r-i} f(o | O_i - o) \\
& \leq \frac{4}{r-i} \sum_{j=i}^r  f(o_j | O_{j+1}) \\
& \le \frac{4f(OPT)}{r-i}.\label{eq:en-upper}
\end{align}

Combining \Cref{eq:en-lower} and \Cref{eq:en-upper} yields
$\E[f(s_i | S_{i-1})] \geq \frac{1}{2\eps} \E[f(o_{i} | O_{i+1})].$
\end{proof}

\subsection{The {\sc ContinuousGreedy} algorithm}
\label{sec:continuousgreedy}

In this section we discuss our implementation of the {\sc ContinuousGreedy} algorithm. The basis of our algorithm is the fast implementation from \cite{BV14}, with additional speed-up due to the fact that the {\sc LazySamplingGreedy+} stage reduces the marginal values of the remaining elements.
The previous section shows that our {\sc LazySamplingGreedy+} algorithm runs with at most $O_\eps(n\log r)$ arithmetic operations, calls to $f$, and calls to the maximum weight data structure. In this section, we describe how {\sc LazySamplingGreedy+} helps the runtime of {\sc ContinuousGreedy}.
%We defer proofs to \Cref{apx:continuousgreedy}.

At the termination of
{\sc LazySamplingGreedy+} it holds that $\tilde{w}(B) \leq \frac{50}{\eps} f(OPT)$.
Stale weights in $B$ have true weights lower than its weight estimate $\tilde{w}_e$. Therefore, the true weight of elements of $B$ must be also at most $\frac{50}{\eps} f(OPT)$. Furthermore, since $B$ is a constant-factor approximation to the true maximum weight basis $B^\star$, this implies that $\textrm{weight}(B^\star) = O(\frac{1}{\eps} f(OPT))$. 

Let $\tilde{f}(T) = f(T|S)$, where $S$ is the set output at the termination of {\sc LazySamplingGreedy+}. We observe that for any set $T\in \cM/S$, $\sum_{e \in T}\tilde{f}(e) = O(\frac{1}{\eps} f(OPT)).$ When this is the case, \cite{BFS17} (Corollary 3.2) gives the following result:

\begin{lemma}[\cite{BFS17}]
{\sc ContinuousGreedy} to obtain a $(1-1/e-\eps)$-approximation uses $O(n\eps^{-2} \log (n/\eps))$ independent set data structure operations and $O(n\eps^{-5} \log^2 (n/\eps))$ queries to $\tilde{f}$.
\end{lemma}

In this section, we make two observations that improve the number of independent set queries by a log factor.  The inner loop of the {\sc ContinuousGreedy} algorithm is essentially a greedy algorithm which operates on a function derived from the multilinear extension of $\tilde{f}$:  $g(T) = F(\bx + \eps \b1_T)$ where $F(\bx) = \E[\tilde{f}(R)]$, $R$ sampled independently with probabilities $x_e$. 
The inner loop of {\sc ContinuousGreedy} finds an increment of the current fractional solution $\bx$ by running a greedy algorithm to approximate a maximum-weight basis with respect to the function $g$. Let us define $w_e(T) = g(T+e) - g(T)$ to be the marginal values of this function.

A fast implementation of this inner loop is the {\sc DescendingThreshold} subroutine of Badanidiyuru and Vondr\'{a}k~\cite{BV14}, which also appears in the algorithm of \cite{BFS17}. 
This subroutine uses the marginal values $w_e(B)$ defined above; the expectation requires $O(\eps^{-1} \log^2(n/\eps))$ samples to estimate for the required accuracy of {\sc ContinuousGreedy}.
In the algorithms below, $w_e(S)$ can be thought of as a black-box that issues $O(\eps^{-1} \log^2(n/\eps))$ calls to the function $\tilde{f}$.

\begin{algbox}[label=alg:desc-thresh]{\sc DescendingThreshold}
$B \leftarrow \emptyset$ \\
$\tau \leftarrow \max_{\{e\} \in \cM} w_e(\emptyset)$ \\
While $\tau \geq \frac{\eps}{r} f(O)$:
\begin{enumerate}
\item Iterate through $e \in E$ one by one. If $B \cup \{e\} \in \cI$ and $w_e(B) \geq \tau$, add $e$ to $B$. Otherwise, if $B \cup \{e\} \notin \cI$, remove $e$ from $E$.
\item $\tau \leftarrow (1-\eps) \tau$
\end{enumerate}
Return $B$
\end{algbox}
The number of independent set queries in {\sc ContinuousGreedy} is dominated by the first line of the while loop in {\sc DescendingThreshold}. 

We make two observations about the {\sc DescendingThreshold} algorithm of Badanidiyuru and Vondr\'{a}k~\cite{BV14}, resulting in two modifications to {\sc DescendingThreshold} that uses the \emph{incremental independence oracle} and \emph{approximate maximum independent set data structure} outlined in \Cref{sec:ds-reqs}.
\begin{observation}\label{obs:incremental}
  Only $O(n/\eps)$ independence oracle queries are required. Furthermore, it is sufficient to use an \bf{\emph{incremental}} independence oracle.
  \end{observation}
  \begin{proof}
   {\sc DescendingThreshold} goes over the elements of $\cM$ repeatedly in different orders, adding elements incrementally to a solution if they are independent from the current solution. Once an element is determined to be incompatible with the current solution, it remains incompatible for the duration of the algorithm. This is because the solution we build is incremental, and incompatibility with the solution implies incompatibility with all supersets of that solution due to matroid properties. Thus the independence oracle is only called $O(n)$ times per run of the {\sc DescendingThreshold} algorithm. Furthermore, the solution of descending threshold greedy is built incrementally, so only an incremental independence oracle is need.
  
  Since {\sc DescendingThreshold} is executed $\eps^{-1}$ times, this implies that $O(n/\eps)$ independence oracle queries are required in total.
  \end{proof}

Thus, we can modify the {\sc DescendingThreshold} of \cite{BV14} by simply ignoring elements that have been previously rejected within the descending threshold greedy subprocedure (see Algorithm~\ref{alg:desc-thresh1}). This yields the following:
\begin{lemma} \label{lem:cg-oracle-calls}
{\sc ContinuousGreedy} uses $O(n/\eps)$ incremental independent set data structure operations and $O(n\eps^{-5} \log^2 (n/\eps))$ queries to $\tilde{f}$.
\end{lemma}

\begin{algbox}[label=alg:desc-thresh1]{\sc DT-Incremental}
$\cD \leftarrow \text{Incremental independence oracle maintaining a set $B$ (\Cref{sec:ds-reqs})}$ \\
$\tau \leftarrow \max_{\{e\} \in \cM} w_e(\emptyset)$ \\
While $\tau \geq \frac{\eps}{r} f(O)$:
\begin{enumerate}
\item $E_\tau \gets \{e \, \mid \, w_e(B) \geq \tau, e \in E \setminus B\}$
\item Iterate through $e \in E_\tau$ one by one. If $\cD.\mbox{\sc Test}(e)$ returns {\sc YES} and $w_e(B) \geq \tau$, call $\cD.\mbox{\sc Insert}(e)$ and add $e$ to $B$. Otherwise, if $B \cup \{e\} \notin \cI$, remove $e$ from $E$.
\item $\tau \leftarrow (1-\eps) \tau$
\end{enumerate}
Return $B$
\end{algbox}

\subsubsection*{An alternative observation}

In the case of transversal matroids, exact incremental independence oracle with polylogarithmic update times are not known.
Instead, we will make the following observation: An approximate {\em decremental maximal independent set} data structure can be used instead of an incremental independence oracle. This results in the modification of descending threshold described in Algorithm~\ref{alg:desc-thresh2}.

\begin{algbox}[label=alg:desc-thresh2]{\sc DT-ApproxIndepSet}
$\cD \leftarrow \text{Approximate dynamic maximum independent set data structure maintaining }
\\ 
\text{\qquad\; a set $B$ (\Cref{sec:ds-reqs})}$  \\
$\tau \leftarrow \max_{\{e\} \in \cM} w_e(\emptyset)$ \\
While $\tau \geq \frac{\eps}{r} f(O)$:
\begin{enumerate}
\item $E_\tau \gets \{e \, \mid \, w_e(B) \geq \tau, e \in E \setminus B\}$
\item $B^+ \gets \cD.\mbox{\sc Batch-Insert}(E_\tau)$
\item While $B^+ \neq \emptyset$: 
\begin{enumerate}
    \item Get any $e \in B^+$. If $w_e(B) < \tau$, $D \gets \cD.\mbox{\sc Delete}(e)$ and set $B^+ \gets B^+ \cup D$.
    \item Remove $e$ from $B^+$.
\end{enumerate}
\item $\tau \leftarrow (1-\eps) \tau$
\end{enumerate}
Return $B$
\end{algbox}

  \begin{observation} \label{obs:decremental}
    Instead of an incremental independence oracle, \textsc{ContinuousGreedy} can be implemented with a stable approximate maximum basis data structure. Furthermore, \textsc{ContinuousGreedy} will only make $O(\eps^{-1}\log r)$ calls to {\sc Batch-Insert} and $O(n\eps^{-1} \log r)$ calls to {\sc Delete}.% that supports only deletions and can be initialized with a subset of elements in the basis.
    \end{observation}
    \begin{proof}
    In {\sc DescendingThreshold}, weights are placed into buckets of geometrically decreasing weights in $(1-\eps)$. For the highest non-empty bucket, we can build a $(1-\eps)$-approximate decremental independent set oracle data structure (with {\sc Batch-Insert}).
    If there is any element  in the approximate maximal independent set of the current bucket which we have not added to our solution yet, we can add it to our solution.
    However, when we recompute the new value of the element, it may have shifted to a lower bucket. 
    If this is the case, we remove the element from our data structure (with {\sc Delete}), move the element to a lower bucket and continue.
    When we are finished with a bucket, we initialize the data structure for the next bucket
    with the subset of elements in the solution so far.
    At the end we may miss adding at most $\eps$ items from each bucket, so this results in a solution
    of $(1-\eps)$ of the total weight.
    \end{proof}

\paragraph{Rounding the fractional solutions.}
The {\sc ContinuousGreedy} algorithm makes $O(1/\eps)$ calls to Algorithm~\ref{alg:desc-thresh}, and outputs a fractional solution that is a convex combination of the $O(1/\eps)$ bases returned by these calls~\cite{BV14}. This fractional solution then needs to be rounded to an integral solution efficiently. In \Cref{sec:rounding}, we show that the data structures we develop can speed up the rounding as well, leading to the overall cost being dominated by the {\sc LazySamplingGreedy+} and {\sc ContinuousGreedy} phases.

\subsection{Analysis of the overall framework}
\begin{lemma}
The approximation returned by our framework has approximation ratio at least $1-1/e-\eps$.
\end{lemma}
\begin{proof}
Let $S_0$ be the set returned by {\sc LazySamplingGreedy+}. Recall that $\tilde{f}(T) := f(T | S_0)$. 
By the results in the previous sections, there exists a set $OPT^\star$ such that $OPT^\star \cup S_0$ is independent and $\E[\tilde{f}(OPT^\star)] \geq (1-\eps/2)f(OPT) - f(S_0)$ (by running {\sc LazySamplingGreedy+} with $\eps /4$ instead of $\eps$). Running continuous greedy on $\tilde{f}$ yields a $(1-1/e-\eps/2)$-approximation $S_1$ to $OPT^\star$. Thus the final value of our solution $f(S_0 + S_1)$ is: 
\begin{align*}
\E[f(S_0 + S_1)] &= \E[\tilde{f}(S_1) + f(S_0)] \\
& \geq (1-1/e-\eps/2) \E[\tilde{f}(OPT^\star) + f(S_0)] \\
& \geq (1-1/e-\eps/2) (1-\eps/2) f(OPT) \\
& \geq (1-1/e-\eps) f(OPT).  \tag*{\qedhere} 
\end{align*}
\end{proof}

\begin{observation}
Our framework uses at most:
\begin{itemize}
    \item $O(n\eps^{-5} \log^2(n/\eps))$ calls to the submodular function oracle $f$.
    \item $O(n\eps^{-1} \log(r/\eps))$ calls to an approximate maximum weight oracle (\Cref{sec:lazysample}).
    \item Either $O(n/\eps)$ incremental oracle data structure operations or $O(\eps^{-1} \log r)$ calls to {\sc Batch-Insert} and $O(n\eps^{-1} \log r)$ calls to {\sc Delete} on a decremental approximate maximum independent set data structure (\Cref{sec:continuousgreedy}).
\end{itemize}
\end{observation}
The cost of evaluating $f$ is dominated by the {\sc ContinuousGreedy} phase (see \Cref{lem:cg-oracle-calls}), as {\sc LazySamplingGreedy+} only uses  $O(n\eps^{-1} \log(r/\eps))$ oracle calls to $f$, where $r$ is the rank of the matroid (\Cref{lem:lg+-oracle-calls}).

%%%%%%%%%%%%%%%%%% MATROID DS %%%%%%%%%%%%%%%%%%%%%%%%

\section{Data structures for various matroids}
In this section, we give dynamic $(c,d)$-approximate maximum weight oracles 
and $(1-\epsilon)$-approximate independence oracles 
for laminar matroids, graphic matroids, and transversal matroids.

\paragraph{Limitations for further improvements.}
For both the laminar, graphic, and transversal matroid, the total runtime of the data structure operations in {\sc LazySamplingGreedy+} and {\sc ContinuousGreedy} is $O_\eps(|\cM|\log^2 |\cM|)$, where $|\cM|$ is the number of matroid elements. Without improving the original {\sc ContinuousGreedy} algorithm itself, it is impossible to improve the runtime further. This is because the {\sc ContinuousGreedy} phase requires at least $O_\eps(|\cM|\log^2 |\cM|)$ oracle calls to $f$, which is at least $O(1)$ cost in any reasonable model of computing. 

\paragraph{Weighted sampling on $(c,d)$-approximate independent sets.}
Our $(c,d)$-approximate maximum weight oracles in \Cref{sec:ds-reqs} require the ability to sample from the independent set they maintain. This sampling operation can be handled independently from the other operations of the data structure, by augmenting the {\sc Decrement} and {\sc Freeze} operations. As this augmentation is the same in all our data structures, we describe it in \Cref{sec:sampling_ds}.

\subsection{Laminar matroids}
Laminar matroids generalize uniform and partition matroids.
In \Cref{sec:laminar} we present a data structure $\cal D$ using top trees~\cite{AHLT05} that maintains a fully dynamic maximum weight basis for a laminar matroid under insertions and deletions of elements with arbitrary weights in $O(\log n)$ update time. This data structure satisfies the $(c,d)$-approximate maximum weight oracle requirements with $c=d=1$ and satisfies the $(1-\eps)$-approximate independence oracle requirements with $\eps=0$.

\paragraph{Dynamic maximum weight oracle.}
The data structure $\mathcal{D}$ maintains the maximum weight basis under insertion and deletions.
For \textsc{Freeze}$(e)$ operations, we don't need to do anything.
For \textsc{Decrement}$(e,w)$ operations, 
we can simulate a decrement with the deletion of $e$ and an insertion of $e$ with the changed weight.
By \Cref{ap:stable}, deleting and inserting the element removes at most the deleted element
and adds at most one element to the maximum  weight basis,
and thus would never remove a frozen element from the basis whose weight never decreases.

\paragraph{Incremental independence oracle.}
This data structure can also be used to implement an incremental independence oracle as follows: Run the data structure $\mathcal{D}$ where every element has the same weight and that maintains a maximum basis $B$.
Both \textsc{Test} and \textsc{Insert} can easily be handled by our data structure.

\subsection{Graphic matroids}

A graphic matroid can be represented with a weighted undirected graph $G = (V, E, w)$ where the weight of 
and edge $e\in E$ is given by $w(e)$.

\paragraph{Dynamic $(1/2,1/2)$-approximate maximum weight oracle.}

To obtain an approximate maximum spanning tree of a graph $G = (V,E)$, 
take the largest edge incident to every vertex, with ties broken according to the edge numbering.
For every vertex $v\in V$, let $E_v$ denote the set of edges incident to $v$. 
We can store the weights of edges in $E_v$ in a heap $H_v$
and maintain that the maximum element of $H_v$ is part of our approximate maximum spanning tree.
It is easy to show that the set of edges maintained, $F$, is a forest with at least $1/2$ the weights and $1/2$ the number
of edges of the optimal maximum spanning tree $T$.

For the correctness of the algorithm we show first that there cannot be any cycle in $F$.
Assume by contradiction that there is a cycle $C$ in $F$. Direct each edge in $C$ towards the vertex where it was the maximum weight edge, breaking ties according to the vertex number.
If $C$ is a cycle, then $C$ must give a directed cycle, where each edge is larger than the next edge in the directed
cycle in the lexicographic order induced by the edge weight and the vertex number. This is a contradiction.

\emph{Approximation factor}. Root $T$ at an arbitrary vertex and consider the vertices of $T$ starting at the leaves.
We will use a simple charging argument to show that $F$ has at least $1/2$ the weight of $T$
and that $|F|\ge |T|/2$.
The edge of a vertex $v$ going to its parent $u$ in the tree $T$ can be charged
to the largest weight edge leaving $v$, which is in $F$. 
Since each edge of $e\in F$ can be charged at most twice from the two endpoints of $e$
by edges of lesser or the same weight, $F$ has at least half the weight of $T$ and at least half the edges as well.

\textsc{Decrement}$(e, w)$: When the weight of an edge $e=(u,v)\in E$ changes to $w$, we update $H_u$ and $H_v$ accordingly.
This may change the maximum weight edge incident to $u$ or $v$,
but we can lookup and accordingly modify our approximate maximum spanning tree
with the new maximum weight edge of $T_u$ and $T_v$ in $O(\log n)$ time and report these changes.

\textsc{Freeze}$(e)$: When we freeze an edge $e = (u,v)\in F$, we can contract the graph along the edge.
To do so, we can merge the heaps $H_u$ and $H_v$ and associate the merged heap with the new merged vertex.
This can be done in $O(\log n)$ time with binomial heaps or $O(1)$ time using the Fibonacci heaps of Fredman and Tarjan \cite{FredmanT87}. When we merge two vertices, the maximum weight edge incident to the new merged vertex may be added to the approximate maximum spanning tree.

\paragraph{Incremental independence oracle.}
Unweighted incremental maximum spanning tree involves checking if inserting any edge increases the
size of the spanning tree. This can be done in $O(\alpha(n))$ update and query time with the disjoint
set union data structure of Tarjan \cite{Tarjan75}.

\subsection{Transversal matroids}
\label{sec:transversal-main}
\paragraph{Representation of transversal matroids.}
As stated in \Cref{sec:preliminaries}, we assume that our transversal matroids are given as minimal representations. This means that the matroid $\cM$ is represented by a bipartite graph $G = ((L, R),E)$ where $|R| = \rank(\cM)$. For sake of notation let $n = |L|$ and $m = |E|$. As a reminder, each element of the matroid corresponds to a node in $L$, and an independent set $I$ is a subset of $L$ such that there exists a matching in $G$ that matches every element of $I$. We will let $N(v)$ denote the neighbors of $v$ in $G$, that is $N(v) = \{u \mid (u,v)\in E\}$.
If $m > n^2$ we can remove neighbours from each vertex in $L$ until their degree is at most $n$. This doesn't affect whether a vertex belongs to an independent set, as it can always be matched. This reduces $m$ to at most $O(n^2)$.

\paragraph{Dynamic $(1-\eps, 1/2)$-approximate maximum weight oracles.}
Recall that in the case of transversal matroids, the weighted setting of \textsc{LazySamplingGreedy+} leads to a dynamic matching problem on a \emph{vertex-weighted} bipartite graph $G=((L,R),E)$, where each vertex $\ell \in L$ has a non-negative weight $w(\ell)$ and all edges incident to $L$ have weight $w(\ell)$.
We assume that each vertex in $L$ has  a value $w_{min} \ge O(w_{max}\eps/n)$ such that we may ignore the weight of any vertex that drops below $w_{min}$. 
For the purposes of \textsc{LazySamplingGreedy+}, we stop if the maximum weight basis decreases below $O(f(OPT)) \ge w_{max}$, and so even if we discard all items with weight less than $O(w_{max}\eps/n)$, we can discard at most an $\eps$ fraction of $f(OPT)$.
Thus after appropriate multiplicative rescaling, we may assume that $w_{min} = 1$ and $w_{max} = (1+\eps)^k$ for $k = O(\log_{1+\epsilon} n)$.
Furthermore we may assume that the weight of any $\ell\in L$ is $(1+\eps)^j$ for some $j\ge 0$ as we can round all weights in the range of $[(1+\eps)^j, (1+\eps)^{j+1}]$ down to the nearest $(1+\eps)^j$ and lose only a $(1+\eps)^{-1}$ factor in the value of the solution.

We will design a data structure that maintains a matching $M$ such that whenever a \textsc{Decrement}$(\ell,w)$ operation is performed on $\ell \in L$, then $\ell$ will be the only node of $L$ that may potentially become unmatched in $M$. 
We will call a data structure that has this property \emph{L-stable}.
The basis we output will be the set of nodes of $L$ matched in $M$.
Note that $L$-stable data structures can handle the \textsc{Freeze} operation by not doing anything and always returning an empty set. No frozen element will be removed from the basis because frozen elements are never decremented. 

The high level idea of our algorithm is as follows:
We want to maintain a maximal matching according to some weights, as this guarantees that at least half as many nodes of $L$ are matched as in the optimum solution.
The question is just which weights to choose and which algorithm to use to guarantee maximality while fulfilling $L$-stability. Note that $L$-stability allows edges in the matching to change, just un-matching a matched vertex of $L$ is forbidden. For this reason we chose an algorithm that is greedy for the vertices in $R$, i.e., each vertex in $R$ is matched with a neighbor of largest weight for a suitable choice of weight. In order to maintain the invariant at every vertex $r$ of $R$ our greedy algorithm allows $r$ to ``steal'' the matched neighbor $l$ of another vertex $r'$ of $R$. This maintains $L$-stability as $l$ remains matched.
However, the newly un-matched vertex $r'$ might want to steal $l$ right back from $r$. To avoid this, we do not use the original weights in the greedy algorithm, but instead we use ``virtual weights'' that are initialized by the original weights and that decrease by a factor of $(1+\epsilon)$ whenever $l$ is (re-)matched. This makes $l$ less attractive for $r'$ and, as $l$ is never re-matched when its weight is below 1, it also guarantees that $l$ is only re-matched $\tilde O_\eps(1)$ times in total over all decrement operations. For formal details and the proof, see \Cref{sec:transversal}.

\begin{theorem} 
Given a bipartite graph $G=((L,R), E)$ and a value $w_{min}$, there exists a L-stable data structure that handles \textsc{Decrement} operations and maintains a $(1-\epsilon,1/2)$-approximate maximum weighted matching provided that the maximum weighted matching has cost at least $w_{min}$.
The total  running time for preprocessing and all operations as well as the total number of changes to the set of matched vertices is $O(|E|(1/\epsilon + \log |L|))$.
Furthermore, the matching maintained is maximal.
\end{theorem}

\paragraph{$(1-\epsilon)$-approximate dynamic maximum independent set data structure.}
Given a bipartite graph $G=(L,R)$ there is a  $(1-\epsilon)$-approximate maximum matching data structure ${\cal D}_M$ for deletions of vertices in $L$ \cite{BLSZ14} which achieves amortized $O(\epsilon^{-1})$ time per delete operation. 
It has three properties that  are crucial for our algorithm: (1) It does not unmatch a previously matched vertex of $L$ as long as it is not deleted, (2) it maintains an explicit integral matching, i.e., it stores at each vertex whether and if so, along which edge it is matched, and (3) the total time for computing the initial matching and all vertex deletions is $O((m + |L|)/\epsilon)$, where $m$ is the number of edges in the initial graph.

Given an initial graph $G_0$ and a partial matching  $B$ of $G_0$ this algorithm can be modified to guarantee that the initial $(1-\epsilon)$-approximate matching computed for $G_0$ matches all vertices of $B \cap L$. See \Cref{sec:transversal} for details.
We use this data structure ${\cal D}_M$ to implement a $(1-\epsilon)$-approximate dynamic maximum independent set data structure for the transversal matroid as follows:

{\sc Batch-Insert}($E'$): Let $B$ be the $(1-\epsilon)$-approximate matching before the update. Initialize a new data ${\cal D}_M$ with all current elements and compute an initial $(1-\epsilon)$-approximate matching computed for $G_0$ matching all vertices in $B \cap L$. This is possible by the discussion above.

{\sc Delete}($e$): Execute a vertex deletion of vertex $e$ in ${\cal D}_M$.

{\sc Test}($e$): Return YES if $e$ is matched and NO otherwise.

\begin{lemma}
Given a transversal matroid there exists a $(1-\epsilon)$-approximate dynamic maximum independent set data structure such that each {\sc Batch-Insert}($B, E_1, E_2$) and all {\sc Delete} operations until the next {\sc Batch-Insert} take  $O((m' + |E_1| + |E_2|)/\epsilon)$ total worst-case time
%,each \emph{Delete($e$)} operation takes time amortized $O(1/\epsilon)$
and each {\sc Test} operation takes $O(1)$ worst-case time.
\end{lemma}

We show here how to slightly modify the algorithm of \cite{BLSZ14} so that it fulfills the following condition (C):
\emph{Given an initial bipartite graph $G_0=((L,R),E)$ and a partial matching $B$ in $G_0$ the algorithm guarantees that the initial $(1-\epsilon)$-approximate matching computed for $G_0$ matches all vertices of $B \cap L$. }

The algorithm for the deletions of vertices in $L$ of~\cite{BLSZ14} is actually based on an algorithm for insertions of  vertices in $R$ that guarantees that there is no augmenting path of length $k := 2 + 2/\epsilon > 2$.
Whenever a vertex $l$ of $L$ is supposed to be deleted, the algorithm inserts a new vertex $r \in R$ with a single edge, namely to $l$. Let us call the set of vertices inserted into $R$ for this purpose $R'$. As the algorithm guarantees that no augmenting path of length 1 exists and $l$ is unmatched, it follows that after the insertion $l$ will be matched to $r$.
Let us denote by $G$ the initial graph without the deleted vertices and without the vertices in $R'$ and by $G'$ the initial graph with the vertices in $R'$.
Lemma VI.5 in~\cite{BLSZ14} shows that  the following inequality holds:
$$|M(G')| - |R'| \ge (1-\epsilon) |M(G)|,$$
i.e.~even when not counting the matched edges incident to the $|R'|$ deleted edges, the matching in $G'$ is a $(1-\epsilon)$-approximation of the optimal matching in $G$.

We now describe how to modify the algorithm to fulfill condition (C).
Recall that in the decremental algorithm $R$ is given initially and never changes, only in the implementation of the decremental algorithm by an incremental algorithm $R$ is augmented by $R'$. Let us denote the initial set $R$ by $R_0$ to avoid confusion. The incremental algorithm in~\cite{BLSZ14} starts with an empty set $R$. Thus, in the implementation of the decremental all vertices of $R_0$ must be inserted during the preprocessing step to compute an initial approximate matching. We will fix the order in which the vertices of $R_0$ are inserted as follows. First all the vertices that are endpoints of an edge in $B$ are inserted and afterwards all remaining vertices of $R_0$ are inserted.

Furthermore, the incremental algorithm maintains a value for every left vertex $l$, called $rank(l)$, which is initially 0. Whenever a vertex is part of an augmenting path, its rank is incremented by 1. After the insertion of a vertex $r \in R$ the algorithm determines a lowest-ranked neighbor $l$ of $r$ and if it is unmatched, matches $l$ with $r$. In the original algorithm if there are multiple neighbors with identical lowest rank, $r$ is free to choose an arbitrary one. To guarantee condition (C) we modify this choice during preprocessing as follows. Order the edges in $B$ in some arbitrary order and process them one after the other in this order.
Let $(l,r)$ be the next  edge with $l \in L$ and $r \in R_0$. Insert $r$ into the data structure. Note that the rank of $l$, and potentially also of the other neighbors of $r$ equals 0. Choose $l$ as the lowest-ranked neighbor of $r$, match $l$ and $r$, and increment the rank of $l$ (and of no other vertex) by 1. After all the edges of $B$ have been processed, insert all remaining vertices of $R_0$ in arbitrary order, matching them as suitable. As the algorithm never unmatches a previously matched edge, all vertices of $B \cap L$ are matched at the end of the preprocessing step, fulfilling condition (C).

\section{Dynamic maximum weight basis for laminar matroids} \label{sec:laminar}
In this section, we describe a dynamic data structure to maintain a maximum weight basis for a laminar matroid. The data structure supports insertions and deletions of 
elements from the matroid and their weights to the matroid in $O(\log n)$ time, where $n$ is the size of the matroid's ground set.

A laminar matroid $\cM$ on a ground set $\cN = 
\{1, 2, \ldots, n\}$ can be described by a tree $\cT$ where:
\begin{itemize}
    \item Each node $v$ in the tree has an associated set $\cS_v$ and a ``capacity'' $C_v$.
    \item The root node has the set $\cN$ and the $n$ leaf nodes have the singleton sets $\{i\}$ for $i=1,2,\ldots,n$. 
    We will slightly abuse notation and refer to a leaf representing the singleton set $\{e\}$ for $e\in \cN$ as the leaf node $e$ and consider the ground set $\cN$ to be the set of leaves.
    \item The sets of the children of a node $v$ form a  partition of $\cS_v$.
\end{itemize}
A set $\cS\subseteq \cN$ is independent if $|\cS\cap \cS_v| \leq C_v$ for every vertex $v \in \cT$. In the case of equality, that is when $|\cS\cap \cS_v| = C_v$, we say that the constraint of $v$ is \emph{tight}. Furthermore for each element $i\in \cN$, we assume that there is a weight $w_i \ge 0$. The maximum weight basis of a laminar matroid is the independent set with the largest sum of weights. 

%To simplify the exposition below, we further assume that the laminar matroid is represented by a full binary tree where each node has at most 2 children. This is done in the following way:
%for nodes $u$ with more than two children, take two arbitrary children, add a new
%dummy node $v$ as the parent of the set with capacity $C_v = \infty $. Repeat until 
%there are at most two children of $u$. %\mh{what about nodes with only 1 child?}

\subsection{Operations supported by the data structure}

Our data structure aims to support the following operations:

\begin{itemize}
    \item 
    \textsc{Insert}$(u, v, w)$:
    Inserts a new element $u$ to existing (non-leaf) $v$ with weight $w$. Note that this will insert $u$ into all sets $\cS_w$ for all vertices $w$ that is an ancestor of $u$ in the tree. If the maximum weight basis changed, report the change.
    
    \item
    \textsc{Delete}$(u)$: 
    Delete of a leaf $u$. If the maximum weight basis changed, report the change.
    
    \item
    \textsc{Query}$()$: Query for the maximum weight basis $B$.
\end{itemize}

For our applications, we know the underlying matroid structure of all elements that would be inserted into the data structure. Thus we can initialize our tree with all non-leaf nodes and their corresponding capacities. 

The data structures we present will maintain an independent set $B$ of maximum weight, and will be \emph{stable}, meaning that the basis $B$ will change by at most one element per insertion or deletion. This is guaranteed by \Cref{ap:stable}. 

The main challenge for our data structures is that finding replacement elements when we delete an element involves both path queries to find when tight capacity constraints that loosened, and subtree queries to find the replacement element. Adding and removing elements also changes capacities along a path. Handling both subtree and path queries in tandem is a non-trivial task.

It is easy to devise slow algorithm running in $O(n)$ per operation which we do in
\Cref{sec:laminar_slow}. 
We  improve this algorithm with heavy-light decomposition by showing how to change the subtree queries into path queries for $O(\log^2 n)$ per operation in 
\Cref{sec:laminar_heavylight}.
Finally we use top trees to handle both path and subtree operations and prove the following theorem in 
\Cref{sec:laminar_top}.

\begin{theorem}
There exists a data structure that maintains the maximum weight independent set of a laminar matroid under insertion and deletions of elements in $O(\log n)$ worst-case time per operation.
\end{theorem}

\paragraph{A simple lower bound.} We briefly remark that this problem is at least as hard as sorting. To sort $n$ elements, one can construct a cardinality matroid (which is a laminar matroid whose tree only consists of the leafs and the root)  on the $n$ elements with root cardinality constraint $1$. Then sorting can be emulated by using $n$ delete element operations. Since sorting takes $\Omega(n \log n)$, each query must take at least $\Omega(\log n)$ time amortized throughout all operations.

\paragraph{Helpful operations}
Our data structure we will implement the some additional operations:
\begin{itemize}
    \item \textsc{Add}$(u)$: Add the element corresponding to leaf node $u$ to the independent set $B$. This element must be able to be added to $u$, i.e. adding it does not violate any laminar constraints.
    \item \textsc{Remove}$(u)$: Remove the element corresponding to leaf node $u$ from $B$. This element must already be in $B$.
    \item \textsc{LowestTightConstraint}$(u)$: For a node $u$, return the node $v$ closest to $u$ on the path from the root of the tree to $u$ with $C_v$ is exactly equal to the number of elements in $B$ that lie in the subtree rooted at $v$.
    \item \textsc{Initialize}$(u,v,w)$: 
    Create a new node $u$ that's a child of $v$ with weight $w$. This only initializes variables in the new node $u$ and may cause $B$ to no longer be the maximum weight basis.
    \item \textsc{Destroy}$(u)$: 
    Destroy a leaf node $u \not \in B$.
    \item \textsc{QueryMin}$(u)$: Return the minimum weight element in the subtree rooted at $u$ that lies in the basis $B$ or return that none exist.
    \item \textsc{QueryMax}$(u)$: Return the maximum weight element in the subtree rooted at $u$ not in the basis $B$ that can be added to $B$ without violating any capacity constraint, or return that none exist.
\end{itemize}
\subsection{A slow dynamic laminar matroid data structure}
\label{sec:laminar_slow}
For the sake of exposition, we first describe a slower data structure that has $O(n)$ update time. Throughout all the updates, we maintain an independent set $B$.
Assume that the laminar matroid $\cM$ initially starts with no elements (and $B = \emptyset$), and $\cM$ constructed be iteratively adding or removing elements.

For each node $v \in \cT$ we maintain the following variables:
\begin{itemize}
    \item $\max_v :=$ maximum weight $e \in \cS_v \setminus B$ that can be added to $B$ (i.e. satisfies matroid constraints) as well as a pointer to an element with that weight in $\cS_v$. If no elements can be added, set $\max_v = -\infty$.
    \item $\min_v :=$ minimum weight $e \in B \cap \cS_v$ as well as a pointer to an element with that weight in $\cS_v$. If $B \cap \cS_v = \emptyset$, set $\min_v = \infty$.
    \item $c_v :=$ the residual capacity of a node $v$, i.e. $C_v - |B \cap \cS_v|$.
\end{itemize}
For $\min_v$ and $\max_v$, we assume that the elements achieving those values are stored as well (with appropriate sentinel values when they don't exist).

The properties of a matroid guarantees that any addition or deletion of a new matroid element will only change $B$ by at most one element by \Cref{ap:stable}. We'll use the variables above to find the correct element in $B$ to replace. We first describe how to update the variables upon matroid insertions / deletions, and the later describe how to update $B$ under these changes.
We first describe the basic operation of updating $B$.

\paragraph{Implementing \textsc{QueryMin}$(u)$ and \textsc{QueryMax}$(u)$}
These values are maintained by $\min_u$ and $\max_u$.

\paragraph{Implementing \textsc{LowestTightConstraint}$(u)$}
We can walk up the tree from $u$ to find the lowest constraint that is tight.

\paragraph{Implementing \textsc{Remove}$(u)$ / \textsc{Add}$(u)$}
We assume that $B$ stays an independent set after removing / adding the element $u$.
To remove / insert an element $u$ from $B$, we simply increment / decrement $C_v$ by 1 for every node $v$ on the path $P$ from $u$ to the root and update the $\min$ and $\max$ values bottom up on $P$ as follows:
Let $x$ be the current node on $P$. If $C_x > 0$, $\max_x$ can be computed by taking the maximum of $\max_y$  and $\max_z$ of the values of the children
$y$  and $z$ of $x$. Otherwise set $\max_x = -\infty$; $\min_x$ can be computed by simply taking the minimum of a node's children. Then the step is repeated with $x$ being set to the parent of $x$ until the root has been processed.
%\paul{Reformat below into Algorithm floats?}
\paragraph{Implementing \textsc{Initialize}$(u,v,w)$}
We create a new leaf node $v$ as a child of $u$.
We set $\max_v$  to $w$, $\min_v$ to $\infty$, and $C_v = 1$. 
This returns $u$.

\paragraph{Implementing \textsc{Destroy}$(u)$}
We simply delete the node $u$.

\paragraph{Implementing \textsc{Insert}$(u,v,w)$}
We begin by creating a new leaf node $u$ with $\textsc{Initialize}(u,v,w)$. 
We check whether $C_x >0$ for all nodes $x$ on the path $P$ from $v$ to the root. If so, we need to add $v$ to $B$.
If, however, $C_x$ is not positive for all nodes $x$ on $P$, there is at least one node $w$ on $P$ for which $C_w = 0$.
Thus we need to determine whether the maximum weight basis needs to change, i.e.~whether there exists an element $v'$ in $B$ that needs to be swapped out for $v$. 
To find the element $v$ swaps out with, we begin with finding $y=\textsc{LowestTightConstraint}(v)$.
If $\min_y < w$, then we 
(1)  remove the element $e$ corresponding to $\min_y$ from $B$ by calling $\textsc{Remove}(e)$.
(2) call $\textsc{Add}(v)$ to add $v$ to the tree.
Note that after (1) it is guaranteed that every node $x$ on $P$ has a positive $c_x$ and, thus, adding $e$ to $B$ and decrementing the $c_x$ values on $P$ accordingly will not create an nodes $x$ with negative $c_x$ value.

\paragraph{Implementing \textsc{Delete}$(u)$}
If the element $u$ is in $B$, we call \textsc{Remove}$(u)$. Now we may delete $u$ from the tree with $\textsc{Delete}(u)$.
Finally, we need to find an additional element to replace $u$ and restore $B$ to a maximum weight basis.
We first find $v = \textsc{LowestTightConstraint}(u)$.
Now we remove $u$ completely from the tree by deleting the leaf node $u$.
Finally we may add the element given by $z=\textsc{QueryMax}(v)$ if it exists, and adding that element to $B$ with $\textsc{Add}(z)$.

\paragraph{Remark} 
In order to maintain the maximum weight base upon $\textsc{Insert}(u,v,w)$ and $\textsc{Delete}(u)$ depend entirely on other functions.
For our later data structures for this we will omit implementing these functions.

\paragraph{Running time analysis}
The time to execute an element insertion or deletion is $O(h)$, where $h$ is the height of the laminar matroid structure.
One issue with this naive data structure is that the laminar matroid structure could be highly unbalanced, i.e. $O(n)$ in height. When this is the case, each update described above may take $O(n)$ time to execute. 

\subsection{Accelerating laminar matroid queries through a heavy light decomposition} \label{sec:laminar_heavylight}

To achieve $O(\log^2 n)$ query time, we will use the heavy-light decomposition of a tree defined by Sleater and Tarjan \cite{SleatorT83}.
Define the $\textrm{size}(v)$ for a node $v$ in the tree to be $|\cS_v|$.
For a node $u$ with children $v$ and $w$, we call the tree edge from 
 $u$ to $v$ is \emph{heavy} if $\textrm{size}(v) > \textrm{size}(w)$ or if $\textrm{size}(v) = \textrm{size}(w)$ and $v < w$.
Otherwise, we call the edge \emph{light}. We will also refer to $v$ as a heavy (light) child of
$u$ if the edge from $v$ to its parent is heavy (light). Note that the binarization procedure from the previous sections guarantee that any non-leaf node has exactly two children with one being heavy and the other being light.

\begin{lemma}{\cite{SleatorT83}}\label{lem:heavy_light}
For any vertex $v$ in tree, there is at most one heavy edge to a child of $v$, and
there are $O(\log n)$ light edges on the path from $v$ to the root.
\end{lemma}

For each node $v\in \cT$ we maintain $c_v$ as before, but instead of storing
$\max_v$, and $\min_v$, we store:

\begin{itemize}
    \item $\maxlight_v:= \max_{u}$ where $u$ is the light child of $v$ (or $w(v)$ if $v$ is a leaf and $v\not \in B$).
    \item $\minlight_v:= \min_{u}$ where $u$ is the light child of $v$ (or $w(v)$ if $v$ is a leaf and $v\in B$).
\end{itemize}

For every heavy chain $H$ in the tree, we store these values in an auxiliary balanced
binary tree that supports the following operations in $O(\log n)$ time:
\begin{enumerate}
    \item Range increments and decrements for $c_v$ for $v\in H$ (i.e. given two nodes $u, v \in H$, increment or decrement all $c_v$ on the path between those nodes by 1).
    \item Reporting the highest / lowest depth node $v\in H$ with $c_v = 0$.
    \item Updating the value of $\maxlight_v$ or $\minlight_v$ for a single $v\in H$.
    \item Query for the maximum of $\maxlight_v$ or minimum of $\minlight_v$ in a contiguous range of the chain.
\end{enumerate}
Note that this allows us to preform the above operations on an arbitrary (not necessarily heavy) path from $u$ to $v$ in $O(\log^2 n)$ time. This can be done by first finding $w$ the lowest common ancestor (LCA) of $u$ and $v$, then updating the $u$ to $w$ path and the $v$ to $w$ path. By Lemma~\ref{lem:heavy_light}, there are at most $O(\log n)$ light edges and $O(\log n)$ heavy paths to update.

\paragraph{Implementing \textsc{QueryMin}$(u)$}
$\min_u$ can be computed as the minimum value of $\minlight_y$ for any node $y$ on the heavy chain that $u$ is on. This is supported by operation 4.

\paragraph{Implementing \textsc{QueryMax}$(u)$}
We need to first find on the heavy chain of $u$ 
the first descendent $z$ with $c_z=0$, which we can do by operation 2.  
Then $\max_v$ is the maximum of $\maxlight_y$ for $y$ on the path from $v$ to $z$ excluding
$z$. Since the path from $v$ to $z$ lies completely in a heavy chain, we 
can use operation 4 to do so in $O(\log n)$ time.

\paragraph{Implementing $\textsc{LowestTightConstraint}(u)$}
This is operation 2. However since $u$ may not be on the same heavy chain as the root, this operation may take $O(\log^2 n)$ time.

\paragraph{Implementing $\textsc{Remove}(u)$ / $\textsc{Add}(u)$}
We assume that $\cB$ stays an independent set after removing / adding the element $u$.
To remove / insert an element $e$ from $\cB$, we simply increment / decrement $c_v$ by 1 for every node from $e$ to the root as range updates (operation 1) and update the $\minlight_v$ and $\maxlight_v$ values on all higher end points of light edges encountered from the path from $u$ up to the root of the tree (operation 3). This takes $O(\log^2 n)$ time overall.

\paragraph{Implementing \textsc{Initialize}$(u,v,w)$}
We create a new leaf node $v$ as a child of $u$.
We set $\maxlight_v$ to $w$, $\minlight_v$ to $\infty$, and $C_v = 1$. 

\paragraph{Implementing \textsc{Destroy}$(u)$}
We delete the node $u$. Since we assume $u \not \in B$, this changes values of $\maxlight_v$ potentially on upper endpoints of light edges on the path from $u$ to the root. We can update these values in $O(\log^2 n)$ time.

\subsection{Using top trees for \texorpdfstring{$O(\log n)$}{O(log n)} time updates}
\label{sec:laminar_top}

The issue with the naive approach described above is that updating paths from a 
leaf to the root may be very costly when the tree is deep. When this happens,
many $\max$ and $\min$ values may change in the tree. 
To get around this we describe a way to accelerate the query times to $O(\log^2 n)$ in \Cref{sec:laminar_heavylight} with heavy-light decomposition.
The main bottleneck of the heavy-light approach is that we need to perform path range query and updates to maintain
certain auxiliary values that take $O(\log^2 n)$ time to perform.
To improve the update and query time to $O(\log n)$ time, we use the top tree interface of Alstrup, Holm, de Lichtenberg, and Thorup~\cite{AHLT05} 
to avoid needing the path range query operation by implicitly storing the query values.
To improve path range updates implicitly necessary for the \textsc{Add} and \textsc{Remove} operations, we 
us a lazy propagation trick to support such updates quickly.
%The top tree interface allows us to avoid this by keeping track of the answer to both possible path range queries with $O(\log n)$ time updates.

\subsubsection{Rooted top trees} 

Typically top trees are an interface for unrooted forests \cite{AHLT05, TarjanW05},
but we will only describe top trees for rooted trees to simplify some of the notation.
Let $T$ be a rooted tree with root $r$.
A top tree over $T$ is a data structure that represents the tree $T$ by \emph{clusters} that represent both a subtree (connected subgraph) and a path of the original tree. 
We denote a cluster representing the path from $u$ to $v$ and subtree $T'$ of $T$ by $(u,v,T')$.
All edges from vertices of $T\setminus T'$ to a vertex in $T'$ must be incident to $u$ or $v$ which are \emph{boundary nodes}.
We will allow leaf nodes to be special boundary nodes and require that each cluster has two boundary points with one of the boundary points as the descendent of the other. 
We will maintain for the sake of notation that for a top tree cluster $L = (u,v,T')$ we will always have $v$ an ancestor of $u$ 
(and thus $v$ is the highest node in $T'\cap T$).
The original edges of the $T$ are \emph{base clusters}.
Top trees contract pairs of clusters to form new clusters, starting from the base clusters, until only one tree clusters. There are two types of valid contractions that we call \emph{join} operations that create a new cluster, and one operation called \emph{split}
that deletes a cluster and returns the two child clusters that the cluster was a join of.
\begin{itemize}
    \item \textsc{JoinCompressCluster}$(L_1, L_2)$ -- Given two clusters $L_1 = (u,v,T_1)$ and $L_2 = (v, w, T_2)$ with $v$ having degree two, combine to form a new cluster $(u, w, T_1\cup T_2)$.
    \item \textsc{JoinRakeCluster}$(L_1, L_2)$ -- Given two clusters $L_1 = (u,x, T_1)$ and $L_2 = (v, x, T_2)$ with $v$ being a leaf of $T$, combine to form a new cluster $(u,x, T_1\cup T_2)$.
    \item \textsc{Split}$(L)$ --  Delete the non-base cluster $L$ and returns the two clusters $L_1$ and $L_2$ that $L$ was a join of.
\end{itemize}
The underlying data structure decides when and what types of join and splits are done to maintain a balanced binary tree.
The top tree interface takes in functions to change internal values of clusters when join and splits are performed,
and only allows a user to have access to clusters through the functions:
\begin{itemize}
    \item \textit{expose}$(u,v)$  -- Modifies the internal structure of the tree with joins and splits to return a root cluster having $u$ and $v$ as the boundary vertices. 
    \item \textit{cut}$(u,v)$ -- Remove the edge $(u,v)$ and return two new root clusters. Since the structure of $T$ never changes, we will always call \emph{cut}$(u, parent(u))$, 
    so we will write \emph{cut}$(u)$ for brevity.
    \item \textit{link}$(u,v)$ -- Add the edge $(u,v)$, similarly
    $v$ will always be the parent of $u$ in $T$, so we will write \emph{link}$(u)$ for brevity.
    
\end{itemize}
The top tree interface guarantees that at most $O(\log n)$ splits and joins are used in its implementation of these user facing functions.
For details on how this is done we refer to the implementation of this top trees interface by 
Tarjan and Warneck \cite{TarjanW05}.

\subsubsection{Rooted top trees for laminar matroids}

Let $T$ be the tree representation for the laminar family with root $r$. 
For each top tree cluster $L = (u, v, T_L)$, we will additionally store the following:
\begin{itemize}
    \item $\minc_L:=$ The minimum $c_v$ value on the path from $u$ to $v$. Note that this value is only correct  if $L$ is an exposed cluster.
    \item $\Delta_L:=$ The change in $c_v$ values on the path from $u$ to $v$ that needs to be lazily propagate to child clusters of $L$ representing a subpath from $u$ to $v$. 
    \item $\argminc_L:=$ The node closest to $u$ on the path from $u$ to $v$ that has $C_w = \minc_L$.
    \item $\maxe_L:=$ The maximum valued leaf of $T'$ reachable from $v$ that is not a descendent of a node with $x\in T'$ with $c_x = 0$. This value is only stored implicitly and will be determined by the following two variables.
    \item $\maxe1_L:=$ The maximum valued leaf of $T'$ that is a descendent of a node on the path from $u$ to $v$
    and not a descendent of a node $x\in T'$ with $c_x = 0$.
    Note that this value does not care if a node $w$ from $u$ to $v$ has $c_w=0$.
    \item $\maxe0_L:=$ The maximum valued leaf of $T'$ that is a descendent of a node on the path from $u$ to $\argminc_L$ (not including $\argminc_L$) and not a descendent of a node $x\in T'$ with $c_x = 0$. 
%    \mh{this is vague and does not seem to fit what you need for $\maxe_L$}
    \item $\mine_L:=$ The minimum valued leaf of $T'$ that is in the independent set $B$.
\end{itemize}
Our definitions of $\maxe0_L$ and $\maxe1_L$ are specifically tailored to maintain the value of $\maxe_L$ regardless of whether the value of $\minc_L$ has changed due to a lazy update.
\[ \maxe_L =
\begin{cases}
\maxe0_L & \text{ if } \minc_L = 0 \\
\maxe1_L & \text{ if } \minc_L > 0 
\end{cases}
\]
We would want that whenever $expose(u,v)$ is called, the value of $\minc_L$ is updated correctly accurate and thus we can compute $\maxe_L$.
This handles the main issue that when a path update occurs,
the value of $\maxe_L$ can change for all clusters representing part of the path that gets updated.

\begin{funcbox}{\sc Initialization}
For every edge $(u,v)\in T$ with $v$ a parent of $u$, initialize a base clusters $L = (u,v, \{(u,v)\})$ 
with the following initial values:
\begin{itemize}
    \item $\minc_L \leftarrow \min\{c_u, c_v\}$.
    \item $\Delta_L \leftarrow 0$.
    \item $\argminc_L \leftarrow \argmin \{c_u, c_v\}$ If there is a tie, set to $c_u$.
    \item $\maxe1_L \leftarrow w_u$ if $u$ is a leaf and null otherwise.
    \item $\maxe0_L \leftarrow w_u$ if $u$ is a leaf and null otherwise.
    \item $\mine_L \leftarrow u$ if $u$ is a leaf corresponding to an element in $B$ and null otherwise.
\end{itemize}
\end{funcbox}

\begin{funcbox}{\textsc{JoinCompressCluster}($L_1 = (u,v,T_1), L_2 = (v, w, T_2)$)}
        Initialize and return a new cluster to $L = (u,w, T_1\cup T_2)$ with the following values:
        \begin{itemize}
            \item $\minc_L \leftarrow \min\{\minc_{L_1} , \minc_{L_2}\}$.
            \item $\Delta_L \leftarrow 0$
            \item $\argminc_L \leftarrow \begin{cases}\argminc_{L_1} \text{ if } \minc_{L_1} \le \minc_{L_2} \\
            \argminc_{L_2} \text{ if } \minc_{L_1} > \minc_{L_2}  \end{cases}$
            \item $\maxe0_L \leftarrow $ maximum valued leaf between $\maxe0_{L_1}$ and $\maxe0_{L_2}$ if $\minc_{L_1} > \minc_{L_2}$, otherwise $\maxe0_{L_1}$.
            \item $\maxe1_L \leftarrow $ maximum valued leaf between $\maxe1_{L_1}$ and $\maxe1_{L_2}$.
            \item $\mine_L \leftarrow $ minimum valued leaf in $B$ between $\mine_{L_1}$ and $\mine_{L_2}$.
        \end{itemize}
\end{funcbox}
\begin{funcbox}{\textsc{JoinRakeCluster}($L_1 = (u,x,T_1), L_2 = (v, x, T_2)$)}
    Initialize and return a new cluster $L = (u, x, T_1 \cup T_2)$ with the values:
        \begin{itemize}
            \item $\minc_L \leftarrow \minc_{L_1}$.
            \item $\Delta_L \leftarrow 0$.
            \item $\argminc_L \leftarrow \arg\min_{L_1}$.
            \item $\maxe0_L \leftarrow $ maximum valued leaf between $\maxe0_{L_1}$ and $\maxe0_{L_2}$ if $\argminc_L \neq x$
                otherwise $\maxe0_L \leftarrow \maxe0_{L_1} \leftarrow \infty$.
            \item $\maxe1_L \leftarrow $ maximum valued leaf between $\maxe1_{L_1}$ and $\maxe1_{L_2}$.
            \item $\mine_L \leftarrow $ minimum valued leaf in $B$ between $\mine_{L_1}$ and $\mine_{L_2}$.
        \end{itemize}
\end{funcbox}
\begin{funcbox}{\textsc{Split}$(L)$}
Get the two child clusters $L_1$ and $L_2$ that $L$ was a join of. 

For $i \in \{1, 2\}$:
\begin{itemize}
\item  $\delta_{L_i} \leftarrow \Delta_{L_i} + \Delta_L$
\item  $\minc_{L_i} \leftarrow \minc_{L_i} + \Delta_L$
\end{itemize}
Delete $L$ and return $L_1$ and $L_2$.

\end{funcbox}

We will say a cluster $L$ is \emph{touched} if it was considered in a join or one of the two children of a split. The following lemma is crucial to the correctness of our data structure and so that $\maxe_L$ can be computed in $O(1)$ time for any cluster returned from an \emph{expose}, \emph{cut}, or, \emph{link}.

\paragraph{Correctness of the data structure.} 
The correctness of the data structure hinges on the following lemma.

\begin{lemma} \label{lem:exposed_cluster}
Whenever a cluster $L$ is considered for an operation, it has the correct value for all values stored in $L$.
\end{lemma}

\begin{proof}[Proof of Lemma~\ref{lem:exposed_cluster}]
Consider the first time $t$ a node $L$ is considered by the algorithm for a split or join since the last time this occurred at the earlier time $t' < t$.
We must have propagated any lazy updates that occurred between time $t'$ and $t$ from its parent.
Furthermore, there can't have been any changes to descendants of $L$ (with insertions or deletions), since any changes would involve splitting $L$ first. Since that structure beneath $L$ did not change since $t'$, and all of $L$'s ancestors have propagated down their lazy updates, we conclude that the values stored at $L$ are correct.
\end{proof}

We can do our desired updates by first exposing the root and the leaf we are updating.
As the implementations of the changes to the auxiliary variables on splits and joins can be done in $O(1)$  time,
this results in an $O(\log n)$ time operations for \emph{expose}, \emph{cut}, and \emph{link}.

\paragraph{Implementing \textsc{QueryMin}$(u)$ and $\textsc{QueryMax}(u)$}
Calling call $cut(u)$ will return a cluster $L$. By Lemma~\ref{lem:exposed_cluster}, all values stored in $L$ are correct. Since the values are correct, we can compute $\maxe_L$ to give the answer for $\textsc{QueryMax}(u)$, and $\mine_L$ gives $\textsc{QueryMin}(u)$.

\paragraph{Implementing $\textsc{LowestTightConstraint}(u)$}
We can begin by letting $L = expose(r, u)$. Now the answer is given by $\argminc_L$.

\paragraph{Implementing $\textsc{Remove}(u)$ / $\textsc{Add}(u)$}
When we've chosen an element $u$ that is a leaf of the tree $T$ to remove from / insert into our basis $B$, we can increment / decrement $c_v$ by $1$ for all $v$ on the root to leaf path by calling $expose(r, u)$ and updating $\Delta_L$ to update the $c_v$ values lazily. In addition, we need to call $cut(u)$ and $link(u)$ in order to update $\mine$ value at $u$.

\paragraph{Implementing $\textsc{Initialize}(u,v, w)$:} 
When we insert an element $u$, we create a new base cluster for the edge from $u$ to its parent $v$ and join it to the tree. The initialization of the edge $(u,v)$ is described earlier.

\paragraph{Implementing $\textsc{Destroy}(u)$:} 
We simply need to call $cut(e)$. 

\section{Transversal matroids}
\label{sec:transversal}

\paragraph{Dynamic $(1-\eps, 1/2)$-approximate maximum weight oracle.}

For every node $\ell \in L$ we will also store a \emph{virtual weight} which we denote by $vw(\ell)$.
Initially $vw(\ell) = w(\ell)$.
Every time $\ell$ becomes matched, we will decrease the virtual weight by a factor of $(1+\eps)^{-1}$. Once its virtual weight is below 1, it will not become unmatched anymore.
If a \textsc{Decrement}$(\ell, w)$ operation happens, we will set $vw(\ell)$ to $w$ if it was larger than $w$, so we maintain that $vw(\ell) \le w(\ell)$ and unmatch $\ell$ if it was matched.
Observe that the algorithm maintains the following invariant for every vertex $\ell$. That the invariant holds can be shown by a simple induction on the number of (re-)match, un-match, and decrement operations of $\ell$:
\begin{invariant}\label{inv1}
For every node $\ell\in L$, $\ell$ is unmatched iff $vw(\ell) = w(\ell)$. 
Otherwise $vw(\ell) \le w(\ell)/(1+\eps)$.
\end{invariant}

Furthermore the algorithm will maintain the invariant that every node $r\in R$ is matched to the neighbor with the highest virtual weight. Formally we state this as follows:
\begin{invariant}\label{inv2}
For every node $r\in R$, if $r$ is matched to $\ell$ with $vw(\ell) = (1+\eps)^j$ for some $j\ge - \lfloor 1/\eps \rfloor$, then for all other $(\ell', r) \in E$, $vw(\ell') \le (1+\eps)^{j+1} = vw(\ell)(1+\eps)$.
\end{invariant}

\noindent
\emph{Data structure.} We maintain the weight of all the matched vertices of $L$. Additionally,
for each vertex $r$ of $R$ we keep the following data structure:
\begin{enumerate}
\item a list $N_r$ of neighbors of $r$, 
\item a pointer $p_r$ into $N_r$ which points to the first element of $N_r$ that has not yet been processed, and 
\item a value $j_r$ indicating that $r$ is currently looking for matches with  virtual weight at least $(1+\epsilon)^{j_r}$, which is initialized to 
$k := \lfloor n/\epsilon \rfloor$.
\end{enumerate}
The reason for maintaining $p_r$ and $j_r$ is as follows:
$r$ greedily tries to match itself with a neighbor of highest virtual weight, i.e. of weight $(1+\epsilon)^{j_r}$ with $j_r= k$. 
To do so it traverses its neighbor list to find the first such neighbor. 
If none exists, it  decreases $j_r$ by one and restarts the traversal from the beginning of its adjacency list, i.e.~it now checks for a neighbor with second highest virtual weight, and this is repeated until a match is found or $j_r$ becomes negative. 
The next time that $r$ is looking for a match, $r$ does not have to recheck the neighbors that it already has checked during this traversal as virtual weights only decrease. Thus, it remembers the next neighbor to be checked with the pointer $p_r$ and restarts its traversal from that point.
As the traversals stop when $j_r$ drops below a threshold of $-\lfloor 1/\eps \rfloor$, this guarantees that the total work (during the whole algorithm) performed by $r$ while looking for a match is $O(|N_r| (k + 1/\eps))$.

We also assume that each vertex has a non-negative identifier and, thus, the condition that no suitable neighbor has been found  so far is checked by testing whether the variable $unmatchedR$, which is initialized to $-1$,  is still negative.
If no match is found after $j_r$ drops below the threshold, then we match $r$ to any unmatched neighbor (which must have weight $0$).

We initialize our matching with the following static algorithm,  \textsc{FindMatching}, that guarantees that the invariants hold after the initialization. This is necessary to guarantee the maximality of the matching.

\begin{algbox}{\textsc{FindMatching}$(G = ((L,R), E)$} \label{alg:FindMatching}
$M \leftarrow \emptyset$. \\
$vw(\ell) \leftarrow w(\ell)$ for all $\ell \in L$.\\
$j_r \leftarrow k$ for all $r\in R$.\\
$N_r \leftarrow$ list of all neighbors of $r$, for all $r\in R$. \\
$p_r \leftarrow$ first element $N_r$.\\
For $r\in R$:
\begin{enumerate}
    \item  \textsc{MatchR}$(r)$.
\end{enumerate}
\end{algbox}

\begin{funcbox}{\sc MatchR($r$)}
\begin{enumerate}
\item $unmatchedR$ = -1
\item While $j_r \ge - \lfloor 1/\epsilon \rfloor$ and $unmatchedR$ is negative:
\begin{enumerate}
    \item While $p_r$ is not NIL and $unmatchedR$ is negative: 
    \begin{enumerate}
        \item Set $\ell$ to the element of   $N_r$ pointed to by $p_r$.
        \item Set $p_r$ to the next element of $N_r$.
        \item If $vw(\ell) \ge (1+\eps)^{j_r}$:
        \begin{enumerate}
            \item $vw(\ell) \leftarrow vw(\ell)/(1+\eps)$
            \item If $\ell$ was matched to $r'$ in $M$:
            \begin{itemize}
                \item Remove $(\ell, r')$ from $M$
                \item Add $(\ell, r)$ to $M$
                \item $unmatchedR$ = $r'$
                %\item \textsc{MatchR}$(r')$
            \end{itemize}
            \item Else:
            \begin{itemize}
                \item Add $(\ell, r)$ to $M$
                \item Return
            \end{itemize}
        \end{enumerate}
    \end{enumerate}
    \item If $p_r$ is NIL:
    \begin{enumerate}
    \item $j_r \leftarrow j_r - 1$
    \item Set $p_r$ to the first element of $N_r$
    \end{enumerate}
%    \item $N_r \leftarrow N(r)$
    \end{enumerate}
\item If $r$ is unmatched, and there is an unmatched neighbor $\ell\in N_r$, match $r$ to $\ell$.
\item If $unmatchedR$ is not -1: \textsc{MatchR}($unmatchedR$)
\end{enumerate}
\end{funcbox}

Observe that \textsc{Matchr}$(r)$ will only match $r$ to nodes that have virtual weight at least $(1+\eps)^{- \lfloor 1/\epsilon \rfloor}$.

When {\sc Decrement}$(\ell, w)$ is called, we may assume we either can round $w$ to the form $(1+\eps)^j$ for some integer $-\lfloor 1/\eps \rfloor \le j \le k$ or is $0$, because the weight would not matter.
The {\sc Decrement} operation is implemented as follows.

\begin{funcbox}{\sc Decrement($\ell$, $w$)}
Let $j$ be the value such that $vw(\ell) = (1+\eps)^j$ \\

If $w < (1+\eps)^{j}$:
\begin{itemize}
    \item $vw(\ell) \leftarrow w$
    \item If $\ell$ is matched to $r\in R$:
    \begin{enumerate}
        \item Remove the edge $(\ell, r)$ from $M$
        \item $\textsc{MatchR}(r)$
        \item Return all changes to the set of matched nodes of $L$.
    \end{enumerate}
\end{itemize}
\end{funcbox}

We next show the desired properties of the algorithm:
\begin{lemma}
The algorithm is $L$-stable.
\end{lemma}
\begin{proof}
Note that once a vertex $l$ of $L$ is matched, it remains matched until the next
{\sc Decrement($l, \cdot$)} operation, as a {\sc MatchR} operation does not un-match any vertex of $L$. It is only during a {\sc Decrement($l, \cdot$)} operation that $l$ is un-matched and it might become matched again during the subsequence {\sc MatchR} operations. As a {\sc Decrement} operation was executed with parameter $l$, it follows that $l$ is not frozen and, thus, the un-matching of $l$ does not violate the $L$-stability of the algorithm.
\end{proof}

\begin{lemma}
The algorithm maintains a maximal cardinality matching.
\end{lemma}
\begin{proof}
By our algorithm either every node of $r$ of $R$ will either be matched, or all neighbors are already matched.
\end{proof}

\begin{lemma}
The algorithm maintains Invariant~\ref{inv2}.
\end{lemma}
\begin{proof}
As the virtual weight of nodes never increases and it remains unchanged while it is matched,
it suffices to show that the invariant holds for $r$ immediately after running \textsc{Matchr}$(r)$ before running {\sc Matchr} on the next unmatched node.
After {\sc Matchr($r$)} has been executed, suppose $r$ is matched to $\ell$ with virtual weight $(1+\eps)^{j}$.

There are two cases to consider.

\emph{Case 1: $j \ge {- \lfloor 1/\eps \rfloor}.$}
Then in Line 2.(a)iii. in {\sc Matchr($r$)} it must have been true that $vw(\ell)
\ge (1+\eps)^{j + 1}$. This implies that $j_r = j+1$.
It follows that $r$ has checked at some earlier point in time that no other neighboring nodes have virtual weight $(1+\eps)^{j_r+1}$ or larger.
As virtual weights never increase, this means that no neighboring node with virtual weight $(1+\eps)^{j_r+1}$ or larger can exist. Thus, all neighbors must have virtual weight at most
$(1+\eps)^{j_r} = (1+\eps)^{j+1}$. 

\emph{Case 2: $j <  {- \lfloor 1/\eps \rfloor}$}.
Then Invariant~\ref{inv2} does not apply to $r$, i.e., it holds trivially.

Since at initialization we call \textsc{Matchr}$(r)$ on all nodes $r\in R$, the invariant holds for all nodes any time after initialization.
\end{proof}

For the running time analysis, we only need to inspect how much work \textsc{MatchR} does for any particular $r\in R$.
It will scan through incident edges at most $\lfloor 1/\eps \rfloor + 1 + k = O(1/\eps + \log n)$ times. 
Thus the total work done by the algorithm throughout the algorithm is at most:
\[ \sum_{r\in R} \sum_{(\ell, r)\in E} O(1/\eps + \log n)  = O(m (1/\eps + \log n)) \]

As we actually maintain the matching explicitely the running time bound also gives an upper bound on the total number of changes to the base.

Next we show an approximation ratio of $(1+\eps)^{-1}$.
The approximation algorithm can be analyzed by inspecting the maximum weight basis $B^*$ and the matching $M^*$ that witnesses the basis, and comparing it to the $M$ we maintain.
We consider the symmetric difference $M\oplus M^* = (M-M^*) \cup (M^* - M)$, 
consisting of cycles, even length paths, and odd length paths. It suffices to analyze the worst case ratio of the weight $M^*$ attains versus the weight $M$ attains on each of these cycles and paths.

The easiest case is for the cycles.
Since the graph is bipartite, the cycle is even, and both $M$ and $M^*$ must match the same endpoints of the matching. 
This means both must match the same set of vertices in $L$ and get the same weight.

For paths we will use the following lemma:
\begin{lemma} \label{lem:decpath}
Let $P$ be a path of $M \oplus M^*$ starting at an unmatched (in $M$) node $\ell^* \in L$ of weight $(1+\eps)^{j^*}$.
%whose last node of positive weight $\ell \in $L$ has
If $P$ has a node $\ell \in L$ with $vw(\ell) = (1+\eps)^{j}$ then $$w(P) - w(\ell^*) \ge \sum_{i=j}^{j^*-1} (1+\eps)^i = \frac{(1+\eps)^{j^*} -(1+\eps)^j}{\eps}.$$
\end{lemma}
For even length paths, if $M$ and $M^*$ share the same set of matched vertices of $L$ on the path, they get the same weight. 
Otherwise, for even length paths $P$ that start at some vertex $\ell^* \in M^*$ with weight $(1+\eps)^{j^*}$ and end at $\ell \in M$ with weight $(1+\eps)^{j}$, we can compute the ratio $w(P\cap M^*) / w(P\cap M)$ by applying Lemma~\ref{lem:decpath}.

\[ \frac{w(P\cap M^*)}{w(P\cap M)} = \frac{w(P) - w(\ell)}{w(P) - w(\ell^*)} 
= 1 + \frac{w(\ell^*) - w(\ell)}{w(P) - w(\ell^*)} 
\le 1 + \eps\frac{ (1+\eps)^{j^*} - (1+\eps)^j}{(1+\eps)^{j^*} - (1+\eps)^j } 
= 1+ \eps 
\]
For an odd length path $P$, by the optimality of $M^*$, the path must start at an unmatched (in $M$) node $\ell^* \in L$ with weight $(1+\eps)^{j^*}$ and end at an unmatched  (in $M$) node $r\in R$.
This means that $M^*$ must contain an  edge $(\ell, r)$ for some $\ell \in L$. Since $r$ is unmatched in $M$, we must have $vw(\ell) < (1+\eps)^{-\lfloor 1/\eps \rfloor}$, i.e.~$j<-\lfloor 1/\eps \rfloor$. 
As $\ell^*$ has weight $(1+\eps)^{j^*}$, it holds that $j^* \ge 0$, which implies $j^* - j \ge 1/\eps$.
We can apply Lemma~\ref{lem:decpath} to $P$ and get:
\[ \frac{w(P\cap M^*)}{w(P\cap M)} = \frac{w(P)}{w(P) - w(\ell^*)} 
= 1 + \frac{w(\ell^*)}{w(P) - w(\ell^*)} 
\le 1 + \frac{ \eps(1+\eps)^{j^*-j}}{(1+\eps)^{j^*-j} - 1 } \]

Claim~\ref{claim1} below shows that
$1 + \frac{ \eps(1+\eps)^{j^*-j}}{(1+\eps)^{j^*-j} - 1 }
\le 1+ \eps + \frac{1}{j^*-j} \le 1 + O(\eps)$, where the last inequality follow from the fact 
that $j^* - j = \Omega( 1/\eps)$

\begin{claim}\label{claim1}
For all integer $k > 0$ and any $\eps \ge 0$ it holds that
$\frac{ \eps(1+\eps)^{k}}{(1+\eps)^{k} - 1 }
\le \eps + \frac{1}{k}$
\end{claim}
\begin{proof}
To show the claim it suffices to show that $\eps(1+\eps)^k \le (\eps + \frac{1}{k}) ((1+\eps)^{k} - 1)$, or equivalently that
$\eps(1+\eps)^k +(\eps + \frac{1}{k}) \le(\eps + \frac{1}{k}) (1+\eps)^{k}$, i.e.,
$\eps + \frac{1}{k} \le \frac{1}{k} (1+\eps)^{k}$.

Note that this inequality holds with equality for $\eps = 0$.
To finish the proof we differentiate both sides by $\eps$ and show that the right side grows as least as fast with $\eps $ as the left side. This implies that the inequality holds for all $\eps \ge 0$.
Differentiating the left side by $\eps$ gives 1, while 
differentiating the right side gives $(1+\eps)^{k-1}$, which is at least 1 as $k\ge 1$. Thus the claim follows.
\end{proof}

\begin{proof}[Proof of Lemma~\ref{lem:decpath}]
Let $\ell_0,r_1, \ell_1, r_2 ..., \ell_{s}, r_{s}$ if $P$ is odd or
$\ell_0,r_1, \ell_1, r_2 ..., \ell_{s}$ if $P$ is even
be the nodes of $P$ in order of the path with $\ell_0 = \ell^*$.
Consider the node $\ell = \ell_{s-1}$ and let $j$ be such that $vw(\ell) = (1+\eps)^j$.
Since $\ell_0$ is unmatched in $M$, we know that $vw(\ell_0) = w(\ell_0) = (1+\eps)^{j^*}$.
By invariant~\ref{inv2} for every $i$ with $1 \ge i \ge s$ on $P$ it holds that $vw(\ell_i) \ge vw(\ell_{i-1})/(1+\eps)$.
Thus 
 by induction it follows that that $vw(\ell_i) \ge (1+\eps)^{j^* - i}$. Since $vw(\ell_{s-1}) = (1+\eps)^j$, we must have $s \ge j^*-j$. Thus:
\[
w(P) - w(\ell^*)
\ge \sum_{i=1}^{s} w(\ell_i) 
\ge \sum_{i=1}^{s} vw(\ell_i) 
\ge \sum_{i=j}^{j^*-1} (1+\eps)^i 
\]
\end{proof}

\appendix

\section*{Appendix}

\section{Insertions and deletions change at most one maximum weight basis element}\label{ap:stable}
In this section, we show that insertions and deletions of weighted elements to a matroid changes the maximum weight basis by at most one element. This lemma is folklore, but we provide a proof here for completeness.

Let $\cM = (E, \cI)$ be a matroid where every element $e\in E$ has a non-negative weight $w(e)$. For sets $S \subseteq E$, we denote the sum of weights in the set by $w(S)$. Let $B$ be the maximum weight basis in $\cM$. An insertion (or deletion) is from $\cM$ involves adding (or removing) an element from $E$, and adding to (removing from) $\cI$ independent sets containing the newly inserted (deleted) element.

\begin{lemma}
Insertion / deletion of an element $e$ changes the maximum-weight basis by at most one element: Either it stays the same, or an element gets replaced by $e$ (in the case of insertion), or $e$ gets replaced by another element (in the case of deletion).
\end{lemma}

\begin{proof}
Without loss of generality, assume all elements in $\cM$ and the newly inserted element have distinct weights. This guarantees that the maximum-weight basis is unique. If some weights are equal, we can break ties arbitrarily but consistently.\footnote{For example, by associating the elements $e \in \cM$ with a unique index $i_e \in \mathbb{Z}^+$ and adding $\eps^{i_e}$ to each elements weight for some small enough $\eps$.} 

Let $B^\prime$ be the maximum weight basis after $e$ is inserted (deleted). If $B = B^\prime$ we are done, so assume this is not the case. We separate the proof into two cases.

\paragraph{Inserting $e$ into $\cM$:}
$e$ must be in $B^\prime \setminus B$, otherwise we have $B'=B$ (since the maximum-weight basis is unique). Suppose for the sake of contradiction that there exists $b^\prime \in B^\prime \setminus B$ such that $b' \neq e$. Hence, $b^\prime$ was in the matroid prior to the insertion of $e$. 

By the strong exchange property, there exists $b \in B \setminus B'$ (which was also in the matroid originally) such that $B - b + b^\prime$ and $B^\prime - b^\prime + b$ are both bases. Since the matroid elements have distinct weights, we must have either $w(B) < w(B - b + b^\prime)$ or $w(B^\prime) < w(B^\prime - b^\prime + b)$. In either case, we get a contradiction: Either $B$ was not a maximum-weight basis originally, or $B^\prime$ is not a maximum-weight basis after the insertion of $e$.

\paragraph{Deleting $e$ from $\cM$:}
If $e \notin B$, then clearly $B'=B$, so assume that $e \in B \setminus B'$.
By the strong exchange property, there is $e' \in B' \setminus B$ such that $B-e+e'$ and $B'+e-e'$ are both bases. We have $w(B-e+e') + w(B'+e-e') = w(B) + w(B')$, so it must be the case that either
$w(B'+e-e') \geq w(B)$ or $w(B-e+e') \geq w(B')$. Unless $B' = B-e+e'$, whichever inequality holds must be strict, by the distinctness of weights. So either $B$ was not a maximum-weight basis originally (since $B'+e-e'$ is a valid basis before deleting $e$) or $B'$ is not a maximum-weight basis after the deletion (since $B-e+e'$ is a valid basis after the deletion). Hence it must be the case that $B' = B-e+e'$.
\end{proof}

\section{Rounding the fractional solutions}
\label{sec:rounding}
In this section, we describe procedures to round the fractional solution output by {\sc ContinuousGreedy} into an integral one. The output of {\sc ContinuousGreedy} is a convex combination of $O(1/\eps)$ bases of the matroid. To round the fractional solution into an integral one, we use the {\sc SwapRounding} algorithm of Chekuri, Vondr{\'{a}}k and Zenklusen~\cite{CVZ10}. This algorithm can be implemented in $O(b\cdot \alpha(\cM))$ time, where $b$ is the number of bases in the fractional solution, and $\alpha(\cM)$ is the time it takes to run the {\sc MergeBases} subroutine (see Algorithm~\ref{alg:mergebases}). Given two bases $B_1$ and $B_2$, the {\sc MergeBases} subroutine creates a new base $B_{12}$ from the elements in $B_1 \cup B_2$. 

\begin{algbox}[label=alg:mergebases]{{\sc MergeBases$(\alpha_1, B_1, \alpha_2, B_2)$}}
$\cS \gets B_1 \setminus B_2$ \\
For each element $i$ in $\cS$:
\begin{itemize}
    \item Find $j \in B_2 \setminus B_1$ such that $B_1 - i + j$ and $B_2 - j + i$ are both bases.
    \item With probability $\alpha_2 / (\alpha_1 + \alpha_2)$, set $B_1 \leftarrow B_1 - i + j$. Otherwise, set $B_2 \leftarrow B_2 - j + i$.
\end{itemize}
Return $B_1$. (At the end of this routine, $B_1 = B_2$.)
\end{algbox}

For {\sc ContinuousGreedy}, $b = O(1/\eps)$, so it does not contribute significantly to the running time. Our goal in this section is to show that that there exists data structures for which $\alpha(\cM)$ is small enough for the rounding phase to not be a bottleneck in the overall algorithm.

\paragraph{Rounding with a transversal matroid constraint}
For a transversal matroid, our prior algorithms provides a matching certifying the bases provided by {\sc ContinuousGreedy}. Given these matchings, $\alpha(\cM)$ can be implemented in $O(\rank^2(\cM))$ time by doing the following:
\begin{enumerate}
    \item Given the matchings $M_1$ and $M_2$ associated with $B_1$ and $B_2$ respectively, form the graph with the edges $M_1 \cup M_2$. This graph consists paths and even length cycles on $O(\rank(\cM))$ edges.
    \item Find any path from a vertex in $B_1 \setminus B_2$ to $B_2 \setminus B_1$. One is guaranteed to exist due to the matroid basis exchange property.
    \item Update $B_1$ or $B_2$ respectively according to the {\sc MergeBases}.
\end{enumerate}
Each augmenting path takes $O(\rank(\cM))$ time to find as there are $O(\rank(\cM))$ edges. This is repeated up to $O(\rank(\cM))$ times, resulting in {\sc MergeBases} running in $O(\rank^2(\cM))$ time. 

\paragraph{Rounding with a laminar matroid constraint}
%\paul{There is an issue with this, as the strong exchange property is not satisfied for arbitrary $i$ that we exchange with the $j$'s.}
Using the data structures developed in \Cref{sec:laminar}, we can implement the choosing of $i$ and $j$ in {\sc MergeBases} in $O(\log n)$ time by doing the following:
\begin{enumerate}
    \item Initialize two copies of the laminar matroid data structure $\cD_1$ and $\cD_2$ with the laminar sets of $\cM$ of \Cref{sec:laminar}, and call \textsc{Insert} on all elements of $B_1 \cup B_2$ into with weight $1$ (we will use both as unweighted data structures).
    In $\cD_1$, for every $e \in B_1$ call \textsc{Add}$(e)$ to update constraints,
    and also call \emph{cut}$(e)$ if $e\in B_1\cap B_2$.
    In $\cD_2$, do the same but for elements $e' \in B_2$.
     
    \item Iterate through the elements of $B_1 \setminus B_2$ one by one, as in the {\sc MergeBases} algorithm. Let the leaf $u_i$ correspond to $i\in B_1$ and $i\notin B_2$ that we are considering in the loop of of the algorithm.
    Let $v = \textsc{LowestTightConstraint}(u_i)$ in $\cD_2$, or the root node if there are no tight constraints. 
    
    Call \textsc{Remove}$(u_i)$ in $\cD_1$ and \emph{cut}$(u_i)$.
    Let $j$ be the matroid element corresponding to leaf $u_j = \textsc{QueryMax}(v)$ in $\cD_1$ which guarantees that if $j$ exists $j\in B_2\setminus B_1$ as we only have elements of $B_1 \cup B_2$ inserted in the data structure and \textsc{QueryMax} guarantees $j$ is not in $B_1$. We guarantee that $i\neq j$ since we cut off $u_i$. This element is guaranteed to exist by the matroid basis exchange property. 
    
    $B_1 - i + j$ is an independent set since calling \textsc{Remove}($i$) in $\cD_1$ means no nodes on the path from $u_i$ to the root is tight, and \textsc{QueryMax} guarantees the path from $v$ to $u_j$ is not tight. 
    $B_2 - j + i$ is an independent set since calling \textsc{Remove}($j$) from $\cD_2$ would make the path from $v$ to the root have no tight nodes, and since $v$ was the lowest tight constraint of $u_i$, this allows us to add $i$ to $B_2-j$.
    
    \item After doing the possible swap (adding $u_i$ back in if we don't swap), we can call \textit{cut}$(u_i)$ and \textit{cut}$(u_j)$ to remove $i$ and $j$ from consideration, without changing the $c_v$ values of the intermediate nodes.
    $i$ and $j$ will still be in their respective (possibly swapped) basis.
    By the matroid basis exchange theorem, they never need to be considered in the future.
\end{enumerate}

\paragraph{Rounding with a graphic matroid constraint}
For a graphic matroid constraint, Ene and Nguyen~\cite{EN19} showed that the {\sc MergeBases} can be implemented in $O(r \log^2 r)$ time with a red-black tree where $r=\rank(\cM)$.

\section{$B\setminus S$ sampling data structure} 
\label{sec:sampling_ds}
Here we describe a data structure $\cS$ used for sampling that we can augment any data structure that maintains an independent set $I$ of a matroid when the weights lie in $k$ distinct buckets $\cB^{(j)}$ for $j=1,...,k$.
For each bucket we will store a list $L_j$ of the elements in that bucket.
For each list $L_j$, we will also store a counter $|L_j|$ to keep track of the number of elements in the list.
In addition we will also keep track of $w(B\setminus S)$.

%\paul{We should describe how to implement this data structure by augmenting the operations of our (c,d) max weight data structures. I wrote a small section in Sec 4 and updated the data structure operations in Sec 3.1 to summarize this. I don't think the exact notation of the $\tilde{w}_e$s are important. We should just describe it in the abstract without references to lazysamplinggreedy. Also, do we need to limit ourselves to store log elements per bucket? We should be able to store all of B minus S right?}
%\david{We don't have log elements per bucket, we can have as many as we want. The only limit is that we have k total buckets, which influences running times}

This data structure supports the following operations:

\begin{itemize}
    \item $\textsc{Add}(e, j)$ - Add $e$ into $L_{j}$. It is guaranteed that $e$ is not in some list $L_i$ before this is called.
    \item $\textsc{Remove}(e)$ - Remove $e$ from the list $L_{j_e}$ it is currently stored in.
    \item $\textsc{Decrement}(e, j)$ - Move $e$ from the current list that it is in to $L_{j}$. $e$ must be already be in some list $L_i$ for $i> j$.
    \item $\textsc{Sample}(t)$ - Return a subset of $ B \setminus S$ with each $e\in B\setminus S$ sampled with probability $p_e = \min(1,\frac{t \cdot w_e}{w(B\setminus S)}) $.
    \item $\textsc{UniformSample}()$ - Return a uniformly random element of $B\setminus S$.
\end{itemize}

\paragraph{Implementation of the data structure}
The operations \textsc{Add}, \textsc{Remove}, and \textsc{Decrement}, only manipulate the lists, and can be done in $O(1)$ time.
To implement \textsc{Sample}$()$, for each list $L_j$, we need to sample each element with probability $p_j = \min(1,\frac{t\cdot w_e}{w(B\setminus S)})$ for any $e\in L_j$.
If $p_j < 1$, we first sample $t_j$ from the binomial distribution $B(|L_j|, p_j)$, then sample $t_j$ elements from $L_j$ without replacement. 
Otherwise, when $p_j = 1$, we simply take all $t_j = |L_j|$ elements.
This whole operation can be done in $O(k + \sum_{j} t_j)$ time.
Implementing \textsc{UniformSample}$()$ can easily be done in $O(1)$.

%%%%%%%%%%%%%%%%%%%%%%%%%%%%%%%%%%%%%%%%%%%%%%%%%%%%%%%%%%%%%%%%%%%%%%%%%%%

\bibliography{ref}

\bibliographystyle{alpha}
\end{document}